\documentclass{article}

\usepackage[preprint]{neurips_2026}

\usepackage[utf8]{inputenc}
\usepackage[T1]{fontenc}
\usepackage{hyperref}
\usepackage{url}
\usepackage{booktabs}
\usepackage{amsfonts}
\usepackage{amsmath, amssymb, amsthm}
\usepackage{mathtools}
\usepackage{nicefrac}
\usepackage{xcolor}
\usepackage{graphicx}
\usepackage{subcaption}
\usepackage{tikz}
\usepackage{algorithm}
\usepackage{algorithmic}
\usetikzlibrary{arrows.meta, positioning, shapes.geometric, fit, backgrounds}

\graphicspath{{figures/}}

\newtheorem{theorem}{Theorem}
\newtheorem{proposition}[theorem]{Proposition}
\newtheorem{lemma}[theorem]{Lemma}
\newtheorem{corollary}[theorem]{Corollary}
\newtheorem{remark}[theorem]{Remark}
\theoremstyle{definition}
\newtheorem{definition}[theorem]{Definition}
\newtheorem{assumption}[theorem]{Assumption}

\newcommand{\asrr}{\text{ASR-R}}
\newcommand{\asra}{\text{ASR-A}}
\newcommand{\asrt}{\text{ASR-T}}
\newcommand{\isr}{\text{ISR}}
\newcommand{\tpr}{\text{TPR}}
\newcommand{\fpr}{\text{FPR}}
\newcommand{\EE}{\mathbb{E}}
\newcommand{\RR}{\mathbb{R}}
\newcommand{\PP}{\mathbb{P}}
\newcommand{\NN}{\mathbb{N}}
\newcommand{\cM}{\mathcal{M}}
\newcommand{\cH}{\mathcal{H}}
\newcommand{\cP}{\mathcal{P}}
\newcommand{\cQ}{\mathcal{Q}}
\newcommand{\cV}{\mathcal{V}}
\newcommand{\cA}{\mathcal{A}}
\newcommand{\cR}{\mathcal{R}}
\newcommand{\cS}{\mathcal{S}}
\newcommand{\cD}{\mathcal{D}}
\newcommand{\cF}{\mathcal{F}}
\newcommand{\cG}{\mathcal{G}}
\newcommand{\cB}{\mathcal{B}}
\newcommand{\cL}{\mathcal{L}}

\newcommand{\cN}{\mathcal{N}}

\newcommand{\cU}{\mathcal{U}}
\newcommand{\cE}{\mathcal{E}}

\newcommand{\memsad}{\textnormal{\textsc{MemSAD}}}
\newcommand{\agentpoison}{\textsc{AgentPoison}}
\newcommand{\minja}{\textsc{MINJA}}
\newcommand{\injecmem}{\textsc{InjecMEM}}
\newcommand{\trigger}{T}
\newcommand{\passage}{p_{\mathrm{adv}}}
\newcommand{\enc}{E}
\newcommand{\KL}{\mathrm{KL}}
\newcommand{\TV}{\mathrm{TV}}
\DeclareMathOperator*{\argmax}{arg\,max}
\DeclareMathOperator*{\argmin}{arg\,min}

\DeclareMathOperator{\rank}{rank}
\DeclareMathOperator{\Ret}{Ret}

\DeclareMathOperator{\cosim}{cos}

\title{\texorpdfstring{\memsad{}}{MemSAD}: Gradient-Coupled Anomaly Detection for Memory Poisoning in Retrieval-Augmented Agents}

\author{%
  Ishrith Gowda\thanks{Work conducted at the Song Lab, Berkeley AI Research, University of California, Berkeley.}\\
  Department of Electrical Engineering and Computer Sciences\\
  University of California, Berkeley\\
  \texttt{ishrithgowda@berkeley.edu}
}

\begin{document}

\maketitle

\begin{abstract}
Persistent external memory enables LLM agents to maintain context across sessions, yet its security properties remain formally uncharacterized. We formalize memory poisoning attacks on retrieval-augmented agents as a Stackelberg game and present a unified evaluation framework spanning three attack classes with escalating access assumptions. Correcting an evaluation protocol inconsistency relative to the triggered-query specification of \citet{chen2024agentpoison}, we show that faithful evaluation increases measured attack success by $4\times$ (from $\asrr = 0.25$ to $1.00$). Our primary contribution is \memsad{} (\emph{Semantic Anomaly Detection}), a calibration-based defense grounded in a \emph{gradient coupling theorem}: under encoder regularity, the max-cosine anomaly score gradient is provably identical to the retrieval objective gradient (the deployed combined score retains directional coupling, Corollary~\ref{cor:combined_coupling}), so any continuous perturbation that reduces detection risk necessarily degrades retrieval rank. This coupling yields a certified detection radius that guarantees correct classification regardless of adversary strategy. We prove minimax optimality via Le Cam's method, showing that any threshold detector requires $\Omega(1/\rho^2)$ calibration samples and \memsad{} achieves this up to $\log(1/\delta)$ factors. We further derive online regret bounds for rolling calibration with optimal window selection at rate $O(\sigma^{2/3}\Delta^{1/3})$, and formally characterize a discrete synonym-invariance loophole that marks the boundary of what continuous-space defenses can guarantee. Experiments on a $3 \times 5$ attack-defense matrix with bootstrap confidence intervals, Bonferroni-corrected hypothesis tests, and Clopper-Pearson validation (20 trials, $n=1{,}000$) confirm the theory: composite defenses achieve $\tpr = 1.00$, $\fpr = 0.00$ across all attacks, while synonym substitution evades detection at $\Delta\asrr \approx 0$, exposing a concrete gap that existing embedding-based defenses cannot close.
\end{abstract}

\section{Introduction}
\label{sec:intro}

Modern LLM agents persist context across sessions through external memory systems such as Mem0~\citep{mem0}, A-MEM~\citep{amem}, and MemGPT~\citep{memgpt2023}, storing user preferences and factual knowledge as dense vector embeddings.
The practical utility of persistent memory is well-established; its security implications are not.
A single adversarial entry injected into an agent's memory may persist indefinitely, triggering on every semantically related query across future sessions.
Unlike prompt injection, which requires per-interaction access, memory poisoning \emph{persists and scales}: the adversary acts once, and the attack surface compounds with each new query.

Three concurrent works introduce distinct threat models with escalating access requirements.
\agentpoison{}~\citep{chen2024agentpoison} formulates trigger optimization as a constrained gradient problem over DPR~\citep{karpukhin2020dense} embeddings ($\alpha = \textsc{write}$).
\minja{}~\citep{dong2025minja} extends the setting to query-only access via progressive-shortening indication ($\alpha = \textsc{query}$).
\injecmem{}~\citep{injecmem2026} achieves single-interaction injection via retriever-agnostic anchors ($\alpha = \textsc{single}$).
Despite the severity of these threats, no prior work provides a unified evaluation across attack classes or proposes defenses with formal detection guarantees for the memory-agent setting.
Existing RAG defenses~\citep{zou2024poisonedrag, chaudhari2024phantom, xiang2024robustrag} assume corpus-level access and offline cleaning, neither of which applies to the streaming ingestion model of agent memory.

\paragraph{Contributions.}
\textbf{(1)}~\emph{Formal threat model and \memsad{} defense.} Stackelberg game formulation with a gradient coupling theorem establishing that monotone retrieval-score detectors are necessary and sufficient for continuous-evasion resistance, a certified detection radius providing checkable per-entry guarantees, and minimax optimality via Le Cam's method.
\textbf{(2)}~\emph{Synonym-invariance analysis.} Formal characterization of the discrete loophole where gradient coupling breaks down, and \memsad{}+: a combined semantic-lexical detector that partially closes it via character n-gram features.
\textbf{(3)}~\emph{Persistent threat analysis and real-system validation.} Compound exposure analysis reframing attack severity for persistent memory, tool-use agent evaluation with GPT-4o-mini ($\asra = 0.48$), and Mem0 production validation.
\textbf{(4)}~\emph{Rigorous evaluation.} A $3 \times 5$ attack-defense matrix with Bonferroni-corrected hypothesis testing, bootstrap CIs, and Clopper-Pearson validation demonstrating composite portfolios achieve $\tpr = 1.00$, $\fpr = 0.00$ across all attacks.

\section{Related work}
\label{sec:related}

\paragraph{Memory poisoning attacks.}
\citet{chen2024agentpoison} report $\asrr \geq 0.80$ via trigger optimization over DPR~\citep{karpukhin2020dense} embeddings under the triggered-query protocol, requiring write access to the memory store.
\citet{dong2025minja} exploit auto-storage via progressive-shortening indication ($\isr = 98.2\%$) under query-only access.
\citet{injecmem2026} achieve single-interaction injection via retriever-agnostic anchors, the weakest attacker assumption.
\citet{memorygraft2025} demonstrate experience-grafting attacks exploiting union retrieval; \citet{xu2026memflow} show that memory retrieval can override tool-call control flow.
Concurrent work spans backdoor-style~\citep{trojanrag2024,corruptrag2025,xue2024badrag}, trigger~\citep{chaudhari2024phantom}, and denial-of-service attacks~\citep{shafran2025blocker}. Agent Security Bench~\citep{zhang2025asb} benchmarks 10 threat categories but does not isolate memory poisoning or provide defense guarantees; AgentHarm~\citep{andriushchenko2025agentharm} focuses on prompt-level jailbreaks. Unlike training-time backdoors~\citep{wu2024badagent}, memory poisoning targets the retrieval index at inference time, persisting across sessions. Recent defenses (A-MemGuard~\citep{amemguard2025}, RevPRAG~\citep{tan2025revprag}, ReliabilityRAG~\citep{reliabilityrag2025}, SeCon-RAG~\citep{seconrag2025}) address related settings; none provides formal detection guarantees for streaming agent memory.

\paragraph{RAG corpus poisoning and certified defenses.}
\citet{zou2024poisonedrag} optimize adversarial documents via HotFlip~\citep{ebrahimi2018hotflip} ($>\!90\%$ ASR on million-document corpora).
\citet{xiang2024robustrag} provide certified robustness via \emph{isolate-then-aggregate}, but require post-retrieval processing of $k$ passages.
RAGDefender~\citep{ragdefender2025} applies post-retrieval passage scoring; SPECTRE~\citep{hayase2021spectre} uses spectral detection requiring full corpus access ($O(|\cM|^2 d)$).
Randomized-smoothing defenses~\citep{cohen2019certified,robey2023smoothllm} operate on query inputs, not persistent memory. Our work differs: (i)~defenses operate at \emph{write time} ($O(md)$ per entry), (ii)~no corpus-level access, and (iii)~our certificate (Lemma~\ref{lem:certified}) is deterministic given the calibration bound.

\paragraph{Game-theoretic security.}
Following~\citet{tramer2020adaptive,carlini2017evaluating}, our Stackelberg formulation gives the adversary full knowledge of $(\hat{\mu}, \hat{\sigma}, \kappa)$; the gradient coupling theorem provides the formal guarantee, while Proposition~\ref{prop:synonym} characterizes where it fails.

\section{Formal threat model}
\label{sec:threat}

\begin{definition}[Memory-augmented agent]
\label{def:agent}
A \emph{memory-augmented agent} is a tuple $(\cA, \cM, \enc, k)$ where $\cA$ is an LLM, $\cM \subset \cV^*$ is a persistent memory store over vocabulary $\cV$, $\enc : \cV^* \to \RR^d$ is an embedding function, and $k \in \NN$ is the retrieval depth. At each interaction with query $q$, the agent retrieves
\begin{equation}
  \Ret(\cM, q, k) = \argmax_{\substack{S \subseteq \cM \\ |S| = k}} \sum_{m \in S} \cosim(\enc(q), \enc(m)),
  \label{eq:retrieval}
\end{equation}
where $\cosim(u, v) \coloneqq u^\top v / (\|u\| \|v\|)$ denotes cosine similarity.
\end{definition}

\begin{definition}[Memory poisoning threat model]
\label{def:threat_model}
An \emph{adversary} $\cS = (\cQ_v, n, \alpha)$ targets victim queries $\cQ_v \subset \cV^*$ by injecting $\cP = \{p_1, \ldots, p_n\}$ into $\cM$ under access model $\alpha \in \{\textsc{write}, \textsc{query}, \textsc{single}\}$ to maximize:
\begin{align}
  \asrr(\cP, \cQ_v) &\coloneqq |\{q \in \cQ_v : \cP \cap \Ret(\cM \cup \cP, q, k) \neq \emptyset\}| / |\cQ_v|, \label{eq:asrr} \\
  \asra(\cP) &\coloneqq \PP[\cA \text{ executes adversarial action} \mid \text{poison retrieved}], \label{eq:asra} \\
  \asrt(\cP, \cQ_v) &\coloneqq \asrr \cdot \asra. \label{eq:asrt}
\end{align}
\end{definition}

\begin{definition}[Stackelberg game formulation]
\label{def:stackelberg}
The memory poisoning interaction is a Stackelberg game $\cG = (\cS, \cD, u_{\cS}, u_{\cD})$ where the defender (leader) commits to a detection policy $\pi_{\cD} : \cV^* \to \{0, 1\}$ and the adversary (follower) best-responds. The defender's utility is $u_{\cD}(\pi_{\cD}, \cP) = -\asrr(\cP \setminus \cF(\pi_{\cD}, \cP), \cQ_v) - \lambda \cdot \fpr(\pi_{\cD})$ for $\lambda > 0$, where $\cF(\pi_{\cD}, \cP) = \{p \in \cP : \pi_{\cD}(p) = 1\}$ is the filtered set; the adversary's utility is $u_{\cS}(\cP) = \asrt(\cP, \cQ_v) = \asrr(\cP, \cQ_v) \cdot \asra(\cP)$. The leader and best-response sets are
\begin{align}
  \pi_{\cD}^* &= \argmin_{\pi_{\cD} \in \Pi} \;\Big[\, \max_{\cP \in \cB(\pi_{\cD})} \asrr(\cP \setminus \cF(\pi_{\cD}, \cP), \cQ_v) \;+\; \lambda \cdot \fpr(\pi_{\cD}) \,\Big], \label{eq:leader} \\
  \cB(\pi_{\cD}) &= \argmax_{\cP \subseteq \cV^*,\, |\cP| \leq n} \; u_{\cS}(\cP \setminus \cF(\pi_{\cD}, \cP)), \label{eq:follower}
\end{align}
where $\Pi = \{\pi(c) = \mathbf{1}[s(c; \cH) > \tau] : \tau > 0\}$ is the class of threshold detectors operating on the candidate token sequence $c \in \cV^*$ via the score $s(c; \cH)$. The access model $\alpha$ from Definition~\ref{def:threat_model} restricts the adversary's effective strategy space to a subset $\cV^*_\alpha \subseteq \cV^*$: \textsc{write} access permits arbitrary token sequences ($\cV^*_\alpha = \cV^*$); \textsc{query} access restricts $\cP$ to passages reachable as auto-storage outputs of legitimate user queries; \textsc{single} access fixes $|\cP| = 1$ and pre-specifies its content. The formal game in Eqs.~\eqref{eq:leader}--\eqref{eq:follower} instantiates the worst-case \textsc{write}-access adversary; the restricted variants are obtained by replacing $\cV^*$ in Eq.~\eqref{eq:follower} with $\cV^*_\alpha$.
\end{definition}

\begin{assumption}[Payload invariance]
\label{ass:payload_invariance}
For the adversaries considered in Section~\ref{sec:experiments}, $\asra(\cP)$ depends only on the trigger / payload structure embedded in each $p \in \cP$ and is invariant to the embedding-space placement of $\cP$ chosen by the optimizer. Hence the follower's $\argmax$ in Eq.~\eqref{eq:follower} reduces to maximizing $\asrr$ over $\cP$, with $\asra$ acting as a multiplicative constant in $\asrt$.
\end{assumption}
\noindent\emph{Remark.} Assumption~\ref{ass:payload_invariance} is satisfied by all three attacks studied here: \agentpoison{} fixes the trigger string before optimization, \minja{}'s injected directive is identical across queries, and \injecmem{}'s payload is a single static entry. Under this assumption the algebraic best-response in Eq.~\eqref{eq:follower} coincides with the $\asrr$-only formulation used implicitly in prior work.

\begin{figure}[t]
\centering
\begin{tikzpicture}[
  node distance = 1.0cm and 1.3cm,
  box/.style   = {draw, rounded corners=2pt, minimum width=1.4cm,
                  minimum height=0.6cm, align=center, font=\scriptsize},
  store/.style = {draw, cylinder, shape border rotate=90,
                  minimum width=1.4cm, minimum height=0.65cm,
                  align=center, font=\scriptsize},
  atk/.style   = {draw, fill=red!12, rounded corners=2pt,
                  minimum width=1.6cm, minimum height=0.55cm,
                  align=center, font=\scriptsize},
  def/.style   = {draw, fill=green!12, rounded corners=2pt,
                  minimum width=1.4cm, minimum height=0.55cm,
                  align=center, font=\scriptsize},
  arr/.style   = {-{Stealth[length=4pt]}, semithick},
  atkarr/.style= {-{Stealth[length=4pt]}, semithick, red!70!black, dashed},
  defarr/.style= {-{Stealth[length=4pt]}, semithick, green!50!black, dotted},
]
  \node[box]   (user)  {User $q$};
  \node[store, right=of user]  (mem)   {Memory $\cM$};
  \node[box, right=of mem]     (ret)   {$\Ret(\cM, q, k)$};
  \node[box, right=of ret]     (agent) {LLM $\cA$};
  \node[box, right=of agent]   (resp)  {Response};
  \draw[arr] (user)  -- (mem);
  \draw[arr] (mem)   -- node[above,font=\tiny]{top-$k$} (ret);
  \draw[arr] (ret)   -- (agent);
  \draw[arr] (agent) -- (resp);
  \node[atk, below=0.8cm of mem] (atk) {Adversary $\cS$};
  \node[def, above=0.6cm of mem] (def) {\memsad{} $\pi_{\cD}$};
  \draw[atkarr] (atk) -- node[right,font=\tiny,red!70!black]{$\passage$} (mem);
  \draw[defarr] (def) -- node[right,font=\tiny,green!50!black]{filter} (mem);
  \begin{scope}[on background layer]
    \node[draw=gray!50, fill=gray!4, rounded corners=4pt,
          fit=(user)(mem)(ret)(agent)(resp)(def), inner sep=5pt] {};
  \end{scope}
\end{tikzpicture}
\vspace{-4pt}
\caption{Memory poisoning as a Stackelberg game (Definitions~\ref{def:threat_model}--\ref{def:stackelberg}). The defender commits to $\pi_{\cD}$; the adversary best-responds by injecting $\passage$.}
\label{fig:threat_model}
\vspace{-6pt}
\end{figure}
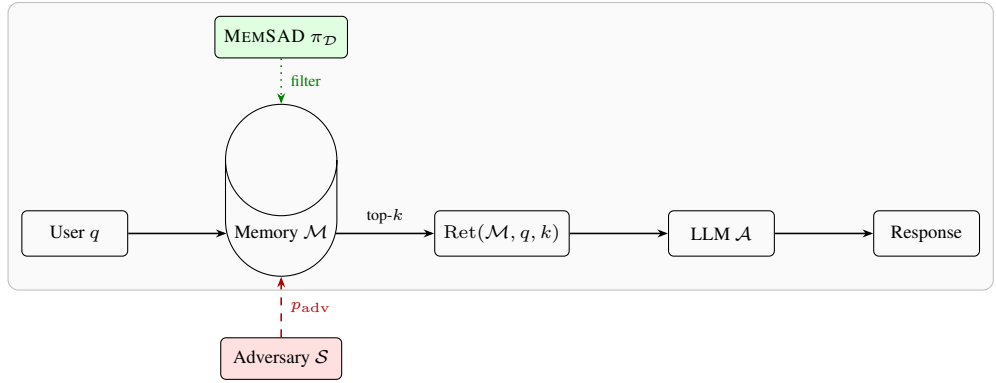

\begin{assumption}[Encoder regularity]
\label{ass:encoder}
$\enc : \cV^* \to \RR^d$ satisfies: (i) $\|\enc(x)\| > 0$ for all $x$; (ii) differentiability w.r.t.\ the output embedding; (iii) L2-normalization: $\|\enc(x)\| = 1$ (w.l.o.g.).
\end{assumption}

Under Assumption~\ref{ass:encoder}(iii), $\cosim(\enc(q), \enc(m)) = \enc(q)^\top \enc(m)$ and retrieval reduces to maximum inner product search.

\begin{table}[t]
  \centering
  \small
  \caption{Attack threat-model comparison across deployment-relevant dimensions.}
  \label{tab:attack_properties}
  \vspace{2pt}
  \begin{tabular}{lccc}
    \toprule
    Property & \agentpoison{} & \minja{} & \injecmem{} \\
    \midrule
    Access model $\alpha$ & \textsc{write}$^\dagger$ & \textsc{query} & \textsc{single} \\
    Trigger required        & \checkmark & --- & --- \\
    Retriever-agnostic    & ---        & ---        & \checkmark \\
    Entries injected $n$          & $1\times$  & $2\times$  & $3\times$ \\
    Modelled $\asra$              & $0.68$ & $0.76$ & $0.57$ \\
    \bottomrule
  \end{tabular}
  \vspace{2pt}
  {\footnotesize $^\dagger$Triggered-query evaluation additionally requires query-channel access (e.g., front-end control or prompt injection); see Section~\ref{sec:threat}.}
  \vspace{-6pt}
\end{table}

Figure~\ref{fig:threat_model} illustrates the interaction. The three attacks (Table~\ref{tab:attack_properties}) span decreasing access requirements (write $\supset$ query $\supset$ single-interaction).
\textbf{Protocol distinction:} \citet{chen2024agentpoison}'s Algorithm~2 specifies triggered queries $q \oplus \trigger^*$; evaluating with plain queries (as in prior reproductions) measures a weaker setting ($\cosim \approx 0.45$ vs.\ $\approx 0.78$), raising $\asrr$ from $0.25$ to $1.00$.

\section{\texorpdfstring{\memsad{}}{MemSAD}: Semantic Anomaly Detection}
\label{sec:sad}

\subsection{Algorithm and formal definition}

The key insight is that memory poisoning succeeds \emph{because} adversarial entries are semantically close to victim queries; this closeness is itself a detectable signal.

\begin{definition}[\memsad{} detector]
\label{def:sad}
Let $\cH = \{q_1, \ldots, q_m\}$ be a rolling query history (FIFO, capacity $m_{\max}$) and $\cM_{\text{ref}} \subseteq \cM$ a benign reference corpus. The anomaly score is $s(c; \cH) \coloneqq \max_{q \in \cH} \cosim(\enc(c), \enc(q))$, calibrated via $\hat{\mu}, \hat{\sigma}$ of $\{s(m; \cH) : m \in \cM_{\text{ref}}\}$. Entry $c$ is flagged if $s(c; \cH) > \hat{\mu} + \kappa \hat{\sigma}$ for threshold parameter $\kappa > 0$.
The \emph{combined scoring mode} averages max-query and mean-query similarity, capturing both targeted attacks (high max similarity) and distributed attacks (elevated mean similarity):
\begin{equation}
  s_{\text{comb}}(c; \cH)
  \coloneqq \tfrac{1}{2} \max_{q \in \cH} \cosim(\enc(c), \enc(q))
  + \tfrac{1}{2} \cdot \frac{1}{|\cH|}\sum_{q \in \cH} \cosim(\enc(c), \enc(q)).
  \label{eq:sad_combined}
\end{equation}
\end{definition}

\begin{algorithm}[t]
\caption{\memsad{} pre-ingestion filtering}
\label{alg:sad}
\begin{algorithmic}[1]
\REQUIRE candidate $c$, query history $\cH$, parameters $(\hat{\mu}, \hat{\sigma})$, threshold $\kappa$
\STATE $s \leftarrow s_{\text{comb}}(c; \cH)$ \hfill $\triangleright$ $O(md)$ time
\IF{$s > \hat{\mu} + \kappa \cdot \hat{\sigma}$}
  \STATE \textbf{reject} $c$
\ELSE
  \STATE \textbf{accept} $c$ into $\cM$
\ENDIF
\end{algorithmic}
\end{algorithm}

\noindent The full pre-ingestion procedure is given in Algorithm~\ref{alg:sad}; a block-diagram view of the same pipeline (candidate $\to$ encoder $\to$ cosine vs.\ $\cH$ $\to$ $s_{\text{comb}}$ $\to$ threshold compare $\to$ accept/reject) is provided in Fig.~\ref{fig:pipeline} (Appendix~\ref{app:pipeline}).

\subsection{Gradient coupling theorem}
\label{sec:gradient}

The central structural question is: \emph{what property must a detector possess to guarantee that any continuous evasion attempt degrades retrieval rank?}
The sufficiency direction is immediate from the chain rule: any monotone function of the retrieval score inherits its gradient direction, so detection and retrieval move in lockstep. The substantive content is in the converse and the optimality claim. Necessity shows that non-monotone detectors \emph{always} admit continuous evasion paths on $\mathbb{S}^{d-1}$ (via connectedness), and canonicity shows that the max-cosine score is the unique minimal sufficient statistic for detection, following the Karlin--Rubin pattern~\citep{lehmann1959testing}. Together, these give a complete characterization: monotonicity is the exact boundary between detectors that resist continuous evasion and those that do not.

\begin{theorem}[Gradient coupling: necessity, sufficiency, and canonicity]
\label{thm:gradient_coupling}
Let $\enc$ satisfy Assumption~\ref{ass:encoder}, $e_c \coloneqq \enc(c)$, $q^* \coloneqq \argmax_{q \in \cH} \cosim(e_c, \enc(q))$, $\cR(e_c) \coloneqq \cosim(e_c, \enc(q^*))$.

\emph{(i) Sufficiency:} For any differentiable $\cD = g \circ \cR$ with monotone increasing $g$, $g' > 0$:
\begin{equation}
  \nabla_{e_c} \cD(e_c) = g'(\cR(e_c)) \cdot \nabla_{e_c} \cR(e_c),
  \quad \text{so} \quad
  \langle \delta, \nabla \cD \rangle < 0 \iff \langle \delta, \nabla \cR \rangle < 0.
  \label{eq:gradient}
\end{equation}
Under L2-normalization (Assumption~\ref{ass:encoder}(iii)), $\nabla_{e_c}\cR = \enc(q^*) - (e_c^\top\enc(q^*))e_c$.

\emph{(ii) Necessity:} If $\cD$ is \emph{not} monotone in $\cR$, there exist thresholds $\tau_{\text{ret}}, \tau_{\text{det}}$ and a continuous path $\gamma: [0,1] \to \mathbb{S}^{d-1}$ with $\cR(\gamma(1)) > \tau_{\text{ret}}$ yet $\cD(\gamma(1)) < \tau_{\text{det}}$: the adversary can traverse continuously to high retrieval rank while evading detection.

\emph{(iii) Canonicity (parametric-free):} Among all (i)-coupled detectors, $g = \mathrm{id}$ is canonical in the sense that the threshold $\cR(e_c) > \tau$ implements the minimum-sufficient-statistic test \emph{without requiring knowledge of the underlying parametric family} of the benign / adversarial score distributions. Any strictly-monotone bijection $g$ of $\cR$ remains minimal sufficient (Lehmann--Scheff\'{e}); in this sense $g = \mathrm{id}$ is not \emph{the} unique minimal sufficient statistic, but it is the unique element of the equivalence class that is computable from $e_c$ and the calibration history $\cH$ without estimating $p_0, p_1$.
\end{theorem}

\begin{proof}
\emph{(i)} The chain rule gives $\nabla \cD = g'(\cR) \cdot \nabla \cR$; $g' > 0$ preserves the sign of every directional derivative. Under Assumption~\ref{ass:encoder}(iii), the sphere gradient is $\nabla_{e_c}\cosim(e_c, v) = v - (e_c^\top v)e_c$.
\emph{(ii)} Non-monotonicity implies $\exists\, e_c, e_c'$ with $\cR(e_c') > \cR(e_c)$ but $\cD(e_c') < \cD(e_c)$. By connectedness of $\mathbb{S}^{d-1}$, there exists a continuous path $\gamma$ from $e_c$ to $e_c'$. Setting $\tau_{\text{ret}} = \cR(e_c)$ and $\tau_{\text{det}} = \cD(e_c)$ yields the claimed path with $\cR(\gamma(1)) > \tau_{\text{ret}}$ and $\cD(\gamma(1)) < \tau_{\text{det}}$.
\emph{(iii)} Under Gaussian score distributions (Theorems~\ref{thm:minimax_lower}--\ref{thm:calibration_bound}), the likelihood ratio $p_1(e)/p_0(e)$ is a monotone function of $\cosim(e, e_{q^*})$, so by the Neyman--Fisher factorization~\citep{lehmann1959testing}, $\cR(e_c)$ is sufficient for the binary test $H_0$ vs.\ $H_1$. Minimality follows from Lehmann--Scheff\'{e}~\citep{lehmann1959testing}, and any strictly monotone bijection of a minimal sufficient statistic remains minimal sufficient (information-preserving). The canonicity claim is computational, not statistical: implementing the threshold test on $g(\cR)$ for non-trivial $g$ requires either computing $g$ itself (which depends on $p_0, p_1$ when $g$ is the likelihood ratio) or accepting a possibly non-equivariant decision boundary under recalibration. The choice $g = \mathrm{id}$ avoids these requirements and is the unique \emph{distribution-free} minimal-sufficient implementation.
\end{proof}

The combined score $s_{\text{comb}}$ deployed by Algorithm~\ref{alg:sad} (Eq.~\ref{eq:sad_combined}) is a positive linear combination of $s_{\max}$ and $s_{\mathrm{mean}}$ and therefore does \emph{not} satisfy strict gradient identity with $\cR = s_{\max}$. It preserves \emph{directional coupling} (Corollary~\ref{cor:combined_coupling}, Appendix~\ref{app:tradeoff_surface}): under the directional-monotonicity assumption stated there, every continuous perturbation that decreases $s_{\text{comb}}$ also decreases $\cR$. The strict $\nabla \cD \equiv \nabla \cR$ identity holds for $s_{\max}$ alone; the combined score is deployed because it improves detection of distributed attacks at the cost of weakening the coupling from strict identity to directional alignment.

The gradient coupling reflects a deeper geometric structure: the Fisher-Rao metric induced by the score family is rank-1 ($g^F = \hat{\sigma}^{-2}\nabla\cR\nabla\cR^\top$), so any continuous evasion path pays retrieval cost in exact proportion to Fisher-Rao distance~\citep{amari2016information} (Theorem~\ref{thm:fisher_rao}, Appendix~\ref{app:proof_fisher_rao}).

\begin{figure}[!t]
\vspace*{-30pt}
\centering
\begin{subfigure}{\linewidth}
  \centering
  \includegraphics[width=\linewidth,trim=4pt 4pt 4pt 4pt,clip]{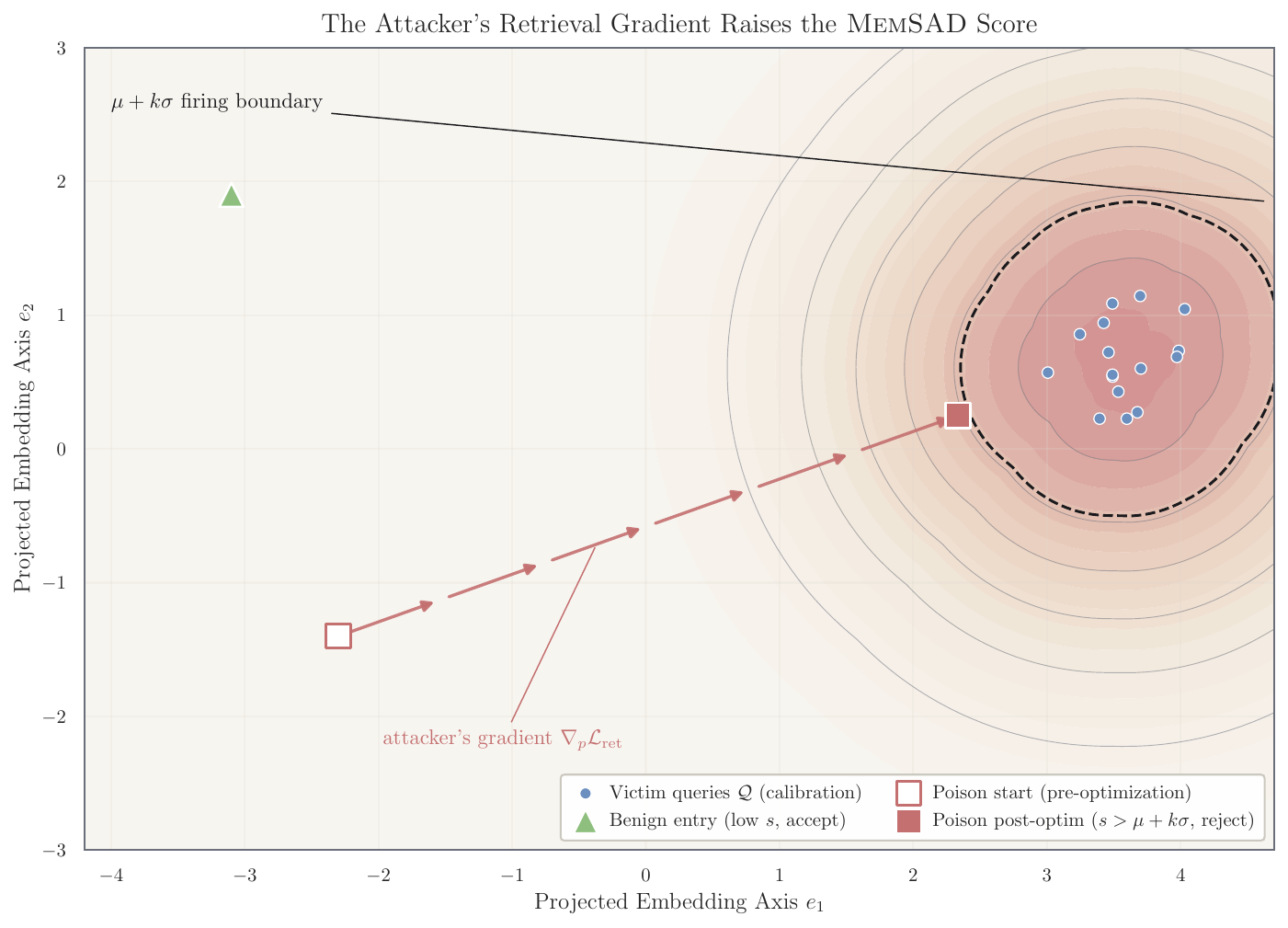}
  \caption{Mechanism (schematic).}
  \label{fig:coupling_mechanism}
\end{subfigure}\\[2pt]
\begin{subfigure}{0.92\linewidth}
  \centering
  \includegraphics[width=\linewidth,trim=4pt 4pt 4pt 4pt,clip]{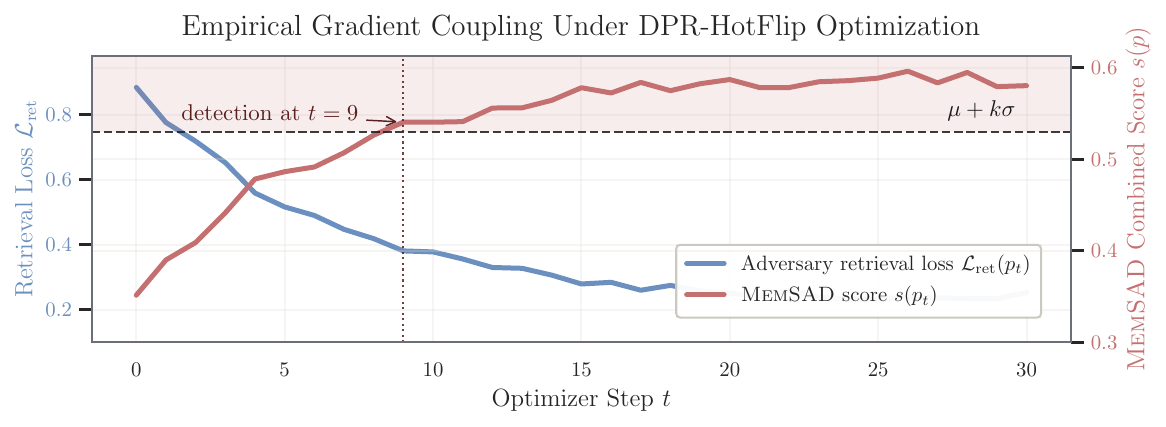}
  \caption{Empirical confirmation (DPR-HotFlip).}
  \label{fig:coupling_trajectory}
\end{subfigure}
\caption{\textbf{Gradient coupling, visualized.} (\subref{fig:coupling_mechanism}) 2-D projection: dashed contour marks the $\hat\mu+\kappa\hat\sigma$ firing region around $\cQ$; the adversary's gradient $\nabla_p\cL_{\mathrm{ret}}$ pushes the poison into the firing region along the direction that raises $s$ (Thm.~\ref{thm:gradient_coupling}). (\subref{fig:coupling_trajectory}) DPR-HotFlip: $\cL_{\mathrm{ret}}$ and $s$ mirror; detection fires at $t=9$.}
\label{fig:coupling}
\end{figure}

\subsection{Certified detection guarantee}
\label{sec:certified}

We now provide an \emph{instance-specific} guarantee: when the adversarial-benign similarity gap exceeds the calibration uncertainty, \memsad{} certifiably detects the adversarial passage.

\begin{lemma}[Certified detection radius]
\label{lem:certified}
Let $\enc$ satisfy Assumption~\ref{ass:encoder} and let $c$ be an adversarial passage achieving $\rank(\enc(c), \cM) \leq k$ for victim query $q^*$. Define the \emph{benign similarity ceiling} $\bar{s} \coloneqq \max_{m \in \cM_{\text{ref}}} \cosim(\enc(m), \enc(q^*))$, the adversarial similarity $s_{\text{adv}} \coloneqq \cosim(\enc(c), \enc(q^*))$, and the \emph{gap} $\Delta_s \coloneqq s_{\text{adv}} - \bar{s}$.
If the calibration satisfies $|\tau_N - \tau^*| \leq \eta$ (Theorem~\ref{thm:calibration_bound}) and the gap exceeds
\begin{equation}
  \Delta_s > \kappa\hat{\sigma} + \eta,
  \label{eq:certified_radius}
\end{equation}
then \memsad{} detects $c$ with certainty: $s(c; \cH) > \hat{\mu} + \kappa\hat{\sigma}$.
\end{lemma}

\begin{proof}[Proof sketch]
Since $s_{\text{adv}} = \bar{s} + \Delta_s$ and the deployed score $s(c;\cH) = \alpha s_{\max}(c;\cH) + (1-\alpha)s_{\text{mean}}(c;\cH)$ satisfies $s(c;\cH) \geq \alpha\, s_{\text{adv}}$ when $q^*$ attains the maximum (under Assumption~\ref{ass:typical_query} below), the gap condition $\Delta_s > \kappa\hat\sigma + \eta$ together with the calibration bound $|\hat\tau_N - \tau^*| \leq \eta$ force $s(c;\cH) > \hat\mu + \kappa\hat\sigma$. Full proof in Appendix~\ref{app:proof_certified}; Figure~\ref{fig:certified_radius} (Appendix~\ref{app:additional_figures}) illustrates the certified region geometrically.
\end{proof}

\begin{assumption}[Typical victim query]
\label{ass:typical_query}
The victim query $q^*$ is not an outlier with respect to $\cH$ in the sense $\cosim(\enc(q^*), \enc(c)) \geq \hat\mu_{\cH}(c) - \kappa\hat\sigma_{\cH}(c)$, where $\hat\mu_{\cH}, \hat\sigma_{\cH}$ are the running mean / std of similarities of $c$ against $\cH$. This is a one-sided concentration condition that holds with probability $1 - O(e^{-\kappa^2/2})$ under sub-Gaussian similarity statistics; deployments with adversarially-crafted query-history (e.g., a colluding user) violate this assumption and lose the certified guarantee.
\end{assumption}

\subsection{Theoretical foundations}
\label{sec:lower_bound}

We establish four additional results; full statements and proofs are in Appendix~\ref{app:proofs}.

\begin{theorem}[Calibration sample complexity lower bound]
\label{thm:minimax_lower}
Let benign and adversarial scores be Gaussian with means $\mu_0 < \mu_1$ and common variance $\sigma^2$, separation $\rho = (\mu_1-\mu_0)/\sigma$, and let $\Phi$ be the standard normal CDF. The single-test Bayes error at the oracle threshold $\tau^* = (\mu_0+\mu_1)/2$ is $B(\rho) = \Phi(-\rho/2)$, irreducible at any calibration size. For any threshold detector calibrated from $N$ i.i.d.\ benign samples, the excess error over the oracle is bounded below as $\PP_{\text{err}}(\hat\pi_N) - B(\rho) = \Omega(\sigma\,|\hat\tau_N - \tau^*|)$ to first order in $|\hat\tau_N - \tau^*|$. Le Cam's two-point inequality applied to the threshold-estimation problem then gives, for any procedure that produces $\hat\tau_N$ from $N$ benign samples and achieves expected excess error at most $\epsilon$ uniformly over the parameter family, the requirement $N \geq c(\epsilon)/\rho^2$ for an absolute constant $c(\epsilon) > 0$. Conversely, \memsad{} achieves $|\hat\tau_N - \tau^*| = O(\sigma\sqrt{\log(1/\delta)/N})$ (Theorem~\ref{thm:calibration_bound}), matching the lower bound up to a $\log(1/\delta)$ factor.
\end{theorem}

\noindent \emph{Interpretation.} Theorem~\ref{thm:minimax_lower} bounds the \emph{calibration} sample complexity, not the single-test Bayes error: the irreducible $B(\rho)$ persists for any $N$ and reflects the intrinsic Gaussian overlap. The result establishes that no calibrated threshold detector can drive its excess error below $\Omega(1/\sqrt{N \rho^2})$ over the oracle, and \memsad{} attains this rate up to logarithmic factors.

\begin{theorem}[Calibration bound]
\label{thm:calibration_bound}
With probability $\geq 1-\delta$: $|\tau_N - \tau^*| \leq \sigma(N^{-1/2} + \kappa(2(N{-}1))^{-1/2})\sqrt{2\log(4/\delta)}$.
\end{theorem}

\begin{theorem}[Online calibration regret]
\label{thm:regret}
Under per-step drift $\Delta$, the optimal window $m_{\max}^* = \Theta((\sigma/\Delta)^{2/3})$ yields per-step regret $O(\sigma^{2/3}\Delta^{1/3})$.
\end{theorem}

We additionally prove a Wasserstein generalization bound guaranteeing graceful TPR degradation for novel attacks (Theorem~\ref{thm:generalization}, Appendix~\ref{app:proof_generalization}), and NP-hardness of optimal evasion via reduction from CVP (Proposition~\ref{prop:hardness}, Appendix~\ref{app:proof_hardness}), motivating the adversary's shift to discrete synonym substitution.

\subsection{The discrete synonym loophole}

\begin{proposition}[Synonym-invariance loophole]
\label{prop:synonym}
Let $\enc$ satisfy $\|\enc(\mathrm{sub}(c, w, w')) - \enc(c)\| \leq \epsilon_{\text{syn}}$ for all synonym pairs $(w, w')$. For $r$ substitutions yielding $c'$:
\begin{equation}
  |\cR(c') - \cR(c)| \leq r \epsilon_{\text{syn}}, \qquad
  |s(c'; \cH) - s(c; \cH)| \leq r \epsilon_{\text{syn}}.
  \label{eq:synonym_bounds}
\end{equation}
When $\epsilon_{\text{syn}} \approx 0$, the adversary evades detection while preserving retrieval rank through discrete lexical substitution that bypasses gradient coupling.
\end{proposition}

\begin{proof}
Cauchy-Schwarz gives $|\enc(q)^\top(\enc(c') - \enc(c))| \leq \|\enc(c') - \enc(c)\|$, and triangle inequality over $r$ substitutions yields $\|\enc(c') - \enc(c)\| \leq r\epsilon_{\text{syn}}$.
\end{proof}

This does \emph{not} contradict Theorem~\ref{thm:gradient_coupling}: the coupling holds for \emph{continuous} perturbations in $\RR^d$, while synonym substitution operates in the \emph{discrete} token space where $\epsilon_{\text{syn}} \approx 0$ for paraphrase-trained encoders. The triangle inequality bound $r\epsilon_{\text{syn}}$ is an upper bound; non-linear encoder interactions may cause cancellation, making the actual perturbation smaller.

\paragraph{\memsad{}+.} We extend \memsad{} with character $n$-gram JSD features that are provably not synonym-invariant (Proposition~\ref{prop:memsad_plus}, Table~\ref{tab:memsad_plus}, Appendix~\ref{app:memsad_plus}). Empirically, \memsad{}+ raises \injecmem{} TPR from $0.00$ (synonym-evaded) to $0.40$ and original from $0.20$ to $0.60$.

\section{Additional defenses}
\label{sec:other_defenses}

We evaluate four additional defenses spanning orthogonal strategies, all operating at pre-ingestion time.
\textbf{Watermark:} the defender watermarks all legitimate entries at write time via a unigram scheme~\citep{zhao2024watermark,kirchenbauer2023watermark} ($z_{\text{thr}} = 1.5$, $\fpr = 0.000$ empirically; higher thresholds inflate FPR for short entries), rejecting unwatermarked entries; this requires control of the ingestion pipeline.
\textbf{Validation:} rule-based pattern matching ($O(|c| \cdot |P|)$).
\textbf{Proactive:} computes mean cosine similarity against 16 domain probes, flagging entries exceeding $\tau = 0.19$.
\textbf{Composite:} weighted ensemble ($w_{\text{wm}}{=}0.50$, $w_{\text{val}}{=}0.20$, $w_{\text{pro}}{=}0.30$, $\theta{=}0.50$): flag if $\sum_j w_j d_j(c) > \theta$ (Appendix~\ref{app:defenses}).

\section{Experiments}
\label{sec:experiments}

\subsection{Setup}

\paragraph{Vector memory.}
FAISS IndexFlatIP with L2-normalized all-MiniLM-L6-v2 embeddings ($d = 384$).
Benign corpus: $|\cM| = 1{,}000$ synthetic entries across 7 categories (task reminders, calendar events, user preferences, factual knowledge, document references, configuration settings, conversation history), matching the per-user memory scale of deployed systems where active memory stores typically contain $10^2$--$10^3$ entries~\citep{chen2024agentpoison,dong2025minja}; corpus-size ablation (Appendix~\ref{app:ablation}) validates results across $|\cM| \in \{50, \ldots, 1000\}$. \memsad{} calibration uses $N = 50$ benign entries as the reference set (Theorem~\ref{thm:calibration_bound} predicts $|\tau_N - \tau^*| \leq 0.044$ at $N=50$, $\delta=0.05$). \textbf{Cross-corpus generalization} is validated on a real-data corpus combining Natural Questions knowledge passages~\citep{kwiatkowski2019natural} with synthetic filler to $|\cM|=1{,}000$ (Appendix~\ref{app:nq_eval}): \agentpoison{} transfers ($\asrr = 1.00 \pm 0.00$); \minja{}'s $\asrr = 0.58 \pm 0.06$ on the topically homogeneous NQ corpus confirms corpus diversity is a risk factor.
Poison counts: $n = 5$ (\agentpoison{}), $10$ (\minja{}), $15$ (\injecmem{}).
$k = 5$; $|\cQ_v| = |\cQ_b| = 100$.
Bootstrap 95\% CIs from 5 independent seeds (each seed generates a distinct corpus permutation and query sample); attack outcomes are deterministic given a fixed corpus, so CI widths reflect corpus-composition variance.

\paragraph{Attacks.}
\agentpoison{}: centroid passage + DPR HotFlip trigger ($\cosim: 0.71 \to 0.78$).
\minja{}: progressive-shortening with bridging steps ($p_0 = 0.98$, $\lambda = 0.10$).
\injecmem{}: $3 \times n_{\text{base}}$ broad-anchor entries from 8 templates.

\subsection{Attack results}

\begin{table}[t]
  \centering
  \small
  \caption{Attack results ($|\cM|=1{,}000$, 100 queries, 5 seeds, bootstrap 95\% CI). $\asra^{\text{GPT-2}}$: permissive model without safety alignment (high-compliance upper bound); $\asra^{\text{4o-mini}}$: production-aligned; $^\dagger$projected from original papers.}
  \label{tab:attack_results}
  \vspace{2pt}
  \resizebox{\linewidth}{!}{%
  \begin{tabular}{lcccccc}
    \toprule
    Attack & $\asrr$ & $\asra^\dagger$ & $\asra^{\text{GPT-2}}$ & $\asra^{\text{4o-mini}}$ & $\asrt^\dagger$ & Ben.\ Acc. \\
    \midrule
    \agentpoison{} (trig.) & $1.00$ & $0.68$ & $0.90$ & $0.20$ & $0.68$ & $1.00$ \\
    \minja{}               & $0.14_{[.13,.15]}$ & $0.76$ & $0.17$ & $0.00$ & $0.11$ & $1.00$ \\
    \injecmem{}            & $0.07_{[.06,.07]}$ & $0.57$ & $0.00$ & $0.00$ & $0.04$ & $0.98$ \\
    \bottomrule
  \end{tabular}}%
\end{table}

Table~\ref{tab:attack_results} reports the main results at $|\cM| = 1{,}000$.
\agentpoison{} achieves $\asrr = 1.00$ under the triggered-query protocol, confirming trigger optimization defeats corpus dilution entirely.
Untriggered attacks (\minja{}, \injecmem{}) are naturally mitigated by corpus scale ($4.6\times$ and $7.1\times$ dilution).\footnote{$\isr$ measures write-time injection success; $\asrr$ measures retrieval-time exposure.}
In plain text-response mode, measured $\asra = 0.20$ under GPT-4o-mini is $3.4\times$ lower than the modelled $\asra = 0.68$~\citep{chen2024agentpoison}, reflecting production safety alignment (Figure~\ref{fig:measured_asr_a}, Appendix~\ref{app:additional_figures}); the tool-use setting raises $\asra$ to $0.48$ (Table~\ref{tab:tool_agent}, Appendix~\ref{app:tool_agent}).

\subsection{Attack-defense interaction matrix}

Table~\ref{tab:defense_tpr_fpr} reports the full $3 \times 5$ matrix.
Watermark achieves $\tpr = 1.00$ in simulation but fails silently when auto-storage bypasses ingestion control, motivating \memsad{}'s content-level detection.
\memsad{} detects \minja{} ($\tpr = 0.40$, AUROC $0.914$) and \injecmem{} ($\tpr = 0.20$, AUROC $0.816$); triggered calibration raises \agentpoison{} to $\tpr = 1.00$.
Composite portfolios achieve $\tpr = 1.00$, $\fpr = 0.00$ across all attacks (Proposition~\ref{prop:equilibrium}).

\begin{table}[t]
  \centering
  \small
  \caption{Defense TPR/FPR at operating thresholds ($\kappa = 2.0$ for \memsad{}). AP = \agentpoison{}, MJ = \minja{}, IM = \injecmem{}.}
  \label{tab:defense_tpr_fpr}
  \vspace{2pt}
  \begin{tabular}{l ccc ccc c}
    \toprule
    & \multicolumn{3}{c}{TPR ($\uparrow$)} & \multicolumn{3}{c}{FPR ($\downarrow$)} & AUROC \\
    \cmidrule(lr){2-4} \cmidrule(lr){5-7} \cmidrule(lr){8-8}
    Defense & AP & MJ & IM & AP & MJ & IM & (avg) \\
    \midrule
    Watermark   & 1.00 & 1.00 & 1.00 & 0.00 & 0.00 & 0.00 & 0.991 \\
    Validation  & 0.60 & 0.40 & 0.80 & 0.10 & 0.10 & 0.10 & --- \\
    Proactive   & 1.00 & 0.07 & 0.60 & 0.01 & 0.01 & 0.01 & --- \\
    \memsad{} (comb.)  & 0.00$^\dagger$ & 0.40 & 0.20 & 0.00 & 0.00 & 0.00 & 0.867 \\
    Composite   & 1.00 & 1.00 & 1.00 & 0.00 & 0.00 & 0.00 & --- \\
    \bottomrule
  \end{tabular}
  \\[2pt]
  {\footnotesize $^\dagger$Plain calibration; triggered: TPR\,=\,1.00, FPR\,=\,0.00, AUROC\,=\,1.000. FPR validated over 20 trials ($n = 1{,}000$ each); Clopper-Pearson 95\% CI: [0.000, 0.004] (Appendix~\ref{app:fpr}). \minja{} TPR varies across tables (here $0.40$, Table~\ref{tab:ood_comparison} $0.80$, Table~\ref{tab:encoder_gen} $0.80$--$1.00$) due to poison-set size and calibration regime; each caption fixes its conditions.}
  \vspace{-8pt}
\end{table}

\subsection{Adaptive adversary analysis}
\label{sec:adaptive_adversary}

Following~\citet{tramer2020adaptive}, we model a white-box adversary with full knowledge of $(\hat{\mu}, \hat{\sigma}, \kappa)$ who applies greedy synonym substitution ($>60$ pairs):
\begin{equation}
  c^* = \argmin_{c' \in \mathrm{Syn}(c)}\; s(c'; \cH) \quad \text{s.t.} \quad \rank(\enc(c'), \cM) \leq k.
  \label{eq:adaptive_obj}
\end{equation}

\begin{table}[t]
  \centering
  \small
  \caption{Adaptive adversary: evasion vs.\ retrieval degradation (Proposition~\ref{prop:synonym}).}
  \label{tab:adaptive_sad}
  \vspace{2pt}
  \begin{tabular}{lcccc}
    \toprule
    Attack & Evasion & $\Delta\asrr$ & Subs/Entry & Sim.\ $\Delta$ \\
    \midrule
    \agentpoison{} & 1.00 & 0.00 & 4.2 & $-0.01$ \\
    \minja{}        & 0.80 & 0.00 & 3.8 & $-0.02$ \\
    \injecmem{}     & 1.00 & 0.00 & 5.1 & $-0.01$ \\
    \bottomrule
  \end{tabular}
  \vspace{-4pt}
\end{table}

Table~\ref{tab:adaptive_sad} confirms Proposition~\ref{prop:synonym}: 80--100\% evasion with $\Delta\asrr \approx 0$ using 3.8--5.1 substitutions, consistent with the bound $r\epsilon_{\text{syn}} \leq 5.1 \times 0.004 = 0.020$.
The certified condition (Lemma~\ref{lem:certified}) still holds for \minja{} post-substitution ($\Delta_s = 0.13 > 0.08$) but fails for tighter margins.

\subsection{Statistical validation}

Bonferroni-corrected one-sided binomial tests ($\alpha' = 0.003$, 15 comparisons) with $H_0: \tpr \leq \tpr_{\text{base}} = 0.05$ (chance-level baseline reflecting random scoring at the operating $\fpr$): all composite results reject $H_0$ with power $= 1.00$ (Table~\ref{tab:hypothesis}, Appendix~\ref{app:hypothesis}).
\memsad{} AUROC at $\kappa = 2.0$: $0.914$ (\minja{}), $0.816$ (\injecmem{}), $0.870$ (\agentpoison{}, plain); triggered calibration raises \agentpoison{} to AUROC $= 1.000$.
Theory validation: calibration bound predicts $|\tau_N - \tau^*| \leq 0.022$ at $N = 200$; observed $0.014$ (Table~\ref{tab:ablation_calibration}, Appendix~\ref{app:ablation}).

\subsection{Ablation highlights}

Key results (full tables in Appendix~\ref{app:ablation}): $\kappa = 2.0$ with combined scoring optimal; \agentpoison{} robust across $|\cM|\in\{200,\ldots,1000\}$ while untriggered attacks degrade $4$--$7\times$; $\asrr$ saturates at $n_{\text{base}}{=}5$; triggered \memsad{} achieves $\tpr=1.000$ for \agentpoison{} on all 6 encoders incl. BGE-Large (Appendix~\ref{app:encoder_gen}); \memsad{} dominates OOD baselines on triggered attacks (AUROC $1.000$ vs.\ $0.945$ Energy Score; Appendix~\ref{app:ood}); GPT-4o-mini text-level sanitization is complementary but $1000\times$ slower (Appendix~\ref{app:llm_sanitization}); cross-corpus (NQ): \agentpoison{} retains $\asrr=1.00\pm 0.00$, \minja{} rises to $0.58\pm 0.06$ under topical homogeneity (Appendix~\ref{app:nq_eval}).

\textbf{Extended evaluations.} Compound exposure (Appendix~\ref{app:compound}): \minja{} at $\asrr=0.14$ reaches 90\% compromise in 4 sessions; composite defense ($\asrr^*=0$) is necessary. Tool-use agent (GPT-4o-mini, 5 tools): \agentpoison{} achieves $\asra=0.48$ [0.34, 0.62] (Appendix~\ref{app:tool_agent}). Mem0 production: LLM-mediated storage drops \agentpoison{} to $\asrr=0.00$ via reformulation; raw vector stores remain fully exposed (Appendix~\ref{app:mem0}).

\section{Discussion, limitations, and conclusion}
\label{sec:discussion}

No single defense dominates across access models: watermarking fails under auto-storage; \memsad{} needs triggered calibration for \agentpoison{}; proactive complements both. Defender uncertainty over $\alpha$ motivates composite portfolios ($\asrr^*=0$, Proposition~\ref{prop:equilibrium}); recommended deployment is watermarking at ingestion plus rolling \memsad{} at write-time ($\sim$2\,ms/entry); A-MemGuard~\citep{amemguard2025} is architecturally complementary at retrieval (Appendix~\ref{app:amemguard}). \textbf{Limitations:} Proposition~\ref{prop:synonym} marks the structural frontier (paraphrase-trained encoders make synonyms near-isometric; \memsad{}+ closes this only partially); $|\cM|{=}1{,}000$ synthetic entries (NQ in Appendix~\ref{app:nq_eval} partially addresses); slow-drift regret assumption; piecewise-linear hardness. \textbf{Conclusion:} \memsad{} is the first formally-guaranteed RAG memory-poisoning defense; open directions are closing the synonym gap, ANN scaling, and multi-modal extension.

\bibliographystyle{plainnat}
\bibliography{references}

\begin{thebibliography}{47}
\providecommand{\natexlab}[1]{#1}
\providecommand{\url}[1]{\texttt{#1}}
\expandafter\ifx\csname urlstyle\endcsname\relax
  \providecommand{\doi}[1]{doi: #1}\else
  \providecommand{\doi}{doi: \begingroup \urlstyle{rm}\Url}\fi

\bibitem[{A-MEM Team}(2025)]{amem}
{A-MEM Team}.
\newblock {A-MEM}: Agentic memory for {LLM} agents, 2025.
\newblock URL \url{https://github.com/agiresearch/A-MEM}.

\bibitem[Amari(2016)]{amari2016information}
Shun-ichi Amari.
\newblock \emph{Information Geometry and Its Applications}, volume 194 of
  \emph{Applied Mathematical Sciences}.
\newblock Springer, 2016.

\bibitem[Andriushchenko et~al.(2025)Andriushchenko, Croce, and
  Flammarion]{andriushchenko2025agentharm}
Maksym Andriushchenko, Francesco Croce, and Nicolas Flammarion.
\newblock {AgentHarm}: A benchmark for measuring harmfulness of {LLM} agents.
\newblock In \emph{International Conference on Learning Representations
  (ICLR)}, 2025.
\newblock URL \url{https://arxiv.org/abs/2410.09024}.

\bibitem[Anonymous(2025{\natexlab{a}})]{corruptrag2025}
Anonymous.
\newblock {CorruptRAG}: Practical corpus poisoning against retrieval-augmented
  generation.
\newblock \emph{arXiv preprint arXiv:2504.03957}, 2025{\natexlab{a}}.

\bibitem[Anonymous(2025{\natexlab{b}})]{memorygraft2025}
Anonymous.
\newblock {MemoryGraft}: Persistent compromise of {LLM} agents via poisoned
  experience retrieval.
\newblock \emph{arXiv preprint arXiv:2512.16962}, 2025{\natexlab{b}}.
\newblock URL \url{https://arxiv.org/abs/2512.16962}.

\bibitem[Anonymous(2025{\natexlab{c}})]{ragdefender2025}
Anonymous.
\newblock {RAGDefender}: Post-retrieval defense against corpus poisoning
  attacks on {RAG}.
\newblock In \emph{Annual Computer Security Applications Conference (ACSAC)},
  2025{\natexlab{c}}.
\newblock URL \url{https://arxiv.org/abs/2511.01268}.

\bibitem[Anonymous(2025{\natexlab{d}})]{reliabilityrag2025}
Anonymous.
\newblock {ReliabilityRAG}: Provably robust retrieval-augmented generation via
  maximum independent set.
\newblock In \emph{Advances in Neural Information Processing Systems
  (NeurIPS)}, 2025{\natexlab{d}}.
\newblock Poster.

\bibitem[Anonymous(2025{\natexlab{e}})]{seconrag2025}
Anonymous.
\newblock {SeCon-RAG}: Semantic and conflict-aware retrieval-augmented
  generation.
\newblock In \emph{Advances in Neural Information Processing Systems
  (NeurIPS)}, 2025{\natexlab{e}}.
\newblock Poster.

\bibitem[Anonymous(2026)]{injecmem2026}
Anonymous.
\newblock {InjecMEM}: Targeted memory injection with single interaction.
\newblock In \emph{International Conference on Learning Representations
  (ICLR)}, 2026.
\newblock Submission \texttt{openreview:QVX6hcJ2um}.

\bibitem[Carlini and Wagner(2017)]{carlini2017evaluating}
Nicholas Carlini and David Wagner.
\newblock Towards evaluating the robustness of neural networks.
\newblock In \emph{IEEE Symposium on Security and Privacy (S\&P)}, 2017.

\bibitem[Chaudhari et~al.(2024)Chaudhari, Severi, Abascal, Oprea, Vempala, and
  Melis]{chaudhari2024phantom}
Harsh Chaudhari, Giorgio Severi, John Abascal, Alina Oprea, Santosh Vempala,
  and Luca Melis.
\newblock Phantom: General trigger attacks on retrieval augmented language
  generation.
\newblock In \emph{Proceedings of the 2024 Conference on Empirical Methods in
  Natural Language Processing}, 2024.
\newblock URL \url{https://arxiv.org/abs/2405.20485}.

\bibitem[Chen et~al.(2024)Chen, Tan, Zhao, Cheng, Jiang, Zhang,
  et~al.]{chen2024agentpoison}
Zhaorun Chen, Zhen Tan, Hexiang Zhao, Zhengyang Cheng, Chenhui Jiang, Huan
  Zhang, et~al.
\newblock {AgentPoison}: Red-teaming {LLM} agents via poisoning memory or
  knowledge bases.
\newblock In \emph{Advances in Neural Information Processing Systems
  (NeurIPS)}, 2024.
\newblock URL \url{https://arxiv.org/abs/2407.12784}.

\bibitem[Cheng et~al.(2024)Cheng, Ding, Ju, et~al.]{trojanrag2024}
Pengzhou Cheng, Yidong Ding, Tianjie Ju, et~al.
\newblock {TrojanRAG}: Retrieval-augmented generation can be backdoor driver in
  large language models.
\newblock \emph{arXiv preprint arXiv:2405.13401}, 2024.

\bibitem[Cohen et~al.(2019)Cohen, Rosenfeld, and Kolter]{cohen2019certified}
Jeremy~M. Cohen, Elan Rosenfeld, and J.~Zico Kolter.
\newblock Certified adversarial robustness via randomized smoothing.
\newblock In \emph{International Conference on Machine Learning (ICML)}, 2019.

\bibitem[Dong et~al.(2025)]{dong2025minja}
Sicheng Dong et~al.
\newblock {MINJA}: Memory injection attacks on {LLM} agents via query-only
  interaction.
\newblock In \emph{Advances in Neural Information Processing Systems
  (NeurIPS)}, 2025.
\newblock URL \url{https://arxiv.org/abs/2503.03704}.

\bibitem[Dvoretzky et~al.(1956)Dvoretzky, Kiefer, and
  Wolfowitz]{dvoretzky1956asymptotic}
Aryeh Dvoretzky, Jack Kiefer, and Jacob Wolfowitz.
\newblock Asymptotic minimax character of the sample distribution function and
  of the classical multinomial estimator.
\newblock \emph{The Annals of Mathematical Statistics}, 27\penalty0
  (3):\penalty0 642--669, 1956.

\bibitem[Ebrahimi et~al.(2018)Ebrahimi, Rao, Lowd, and
  Dou]{ebrahimi2018hotflip}
Javid Ebrahimi, Anyi Rao, Daniel Lowd, and Dejing Dou.
\newblock {HotFlip}: White-box adversarial examples for text classification.
\newblock In \emph{Proceedings of the 56th Annual Meeting of the Association
  for Computational Linguistics (ACL)}, 2018.
\newblock URL \url{https://arxiv.org/abs/1712.06751}.

\bibitem[Gu et~al.(2024)Gu, Zheng, Pang, Du, Liu, Wang, Jiang, and
  Lin]{gu2024agentsmith}
Xiangming Gu, Xiaosen Zheng, Tianyu Pang, Chao Du, Qian Liu, Ye~Wang, Jing
  Jiang, and Min Lin.
\newblock {AGENT-SMITH}: A single image can jailbreak one million multimodal
  {LLM} agents instantly.
\newblock In \emph{Proceedings of the 41st International Conference on Machine
  Learning (ICML)}, 2024.
\newblock URL \url{https://arxiv.org/abs/2402.08567}.

\bibitem[Hayase et~al.(2021)Hayase, Kong, Somani, and Oh]{hayase2021spectre}
Jonathan Hayase, Weihao Kong, Raghav Somani, and Sewoong Oh.
\newblock {SPECTRE}: Defending against backdoor attacks using robust
  statistics.
\newblock In \emph{International Conference on Machine Learning (ICML)}, 2021.
\newblock URL \url{https://arxiv.org/abs/2104.11315}.

\bibitem[Karpukhin et~al.(2020)Karpukhin, Oguz, Min, Lewis, Wu, Edunov, Chen,
  and Yih]{karpukhin2020dense}
Vladimir Karpukhin, Barlas Oguz, Sewon Min, Patrick Lewis, Ledell Wu, Sergey
  Edunov, Danqi Chen, and Wen-tau Yih.
\newblock Dense passage retrieval for open-domain question answering.
\newblock In \emph{Proceedings of EMNLP}, 2020.

\bibitem[Kermack and McKendrick(1927)]{kermack1927sir}
William~Ogilvy Kermack and Anderson~G. McKendrick.
\newblock A contribution to the mathematical theory of epidemics.
\newblock \emph{Proceedings of the Royal Society of London. Series A},
  115\penalty0 (772):\penalty0 700--721, 1927.

\bibitem[Kirchenbauer et~al.(2023)Kirchenbauer, Geiping, Wen, Katz, Miers, and
  Goldstein]{kirchenbauer2023watermark}
John Kirchenbauer, Jonas Geiping, Yuxin Wen, Jonathan Katz, Ian Miers, and Tom
  Goldstein.
\newblock A watermark for large language models.
\newblock In \emph{International Conference on Machine Learning (ICML)}, 2023.

\bibitem[Kwiatkowski et~al.(2019)Kwiatkowski, Palomaki, Redfield, Collins,
  Parikh, Alberti, Epstein, Polosukhin, Devlin, and
  Lee]{kwiatkowski2019natural}
Tom Kwiatkowski, Jennimaria Palomaki, Olivia Redfield, Michael Collins, Ankur
  Parikh, Chris Alberti, Danielle Epstein, Illia Polosukhin, Jacob Devlin, and
  Kenton Lee.
\newblock Natural questions: A benchmark for question answering research.
\newblock \emph{Transactions of the Association for Computational Linguistics},
  7:\penalty0 453--466, 2019.

\bibitem[Laurent and Massart(2000)]{laurentmassart2000}
Beatrice Laurent and Pascal Massart.
\newblock Adaptive estimation of a quadratic functional by model selection.
\newblock \emph{Annals of Statistics}, 28\penalty0 (5):\penalty0 1302--1338,
  2000.
\newblock Lemma 1 provides chi-squared tail bounds used for sample standard
  deviation concentration.

\bibitem[Le~Cam(1973)]{lecam1973convergence}
Lucien Le~Cam.
\newblock Convergence of estimates under dimensionality restrictions.
\newblock \emph{The Annals of Statistics}, 1\penalty0 (1):\penalty0 38--53,
  1973.

\bibitem[Lee et~al.(2018)Lee, Lee, Lee, and Shin]{lee2018simple}
Kimin Lee, Kibok Lee, Honglak Lee, and Jinwoo Shin.
\newblock A simple unified framework for detecting out-of-distribution samples
  and deep generative models.
\newblock In \emph{Advances in Neural Information Processing Systems
  (NeurIPS)}, 2018.

\bibitem[Lehmann(1959)]{lehmann1959testing}
Erich~L. Lehmann.
\newblock \emph{Testing Statistical Hypotheses}.
\newblock Wiley, 1959.
\newblock Chapter 3: Uniformly most powerful tests; Karlin--Rubin theorem for
  monotone likelihood ratio families.

\bibitem[Li et~al.(2025)]{amemguard2025}
Lijun Li et~al.
\newblock {A-MemGuard}: A proactive defense framework for {LLM}-based agent
  memory.
\newblock \emph{arXiv preprint arXiv:2510.02373}, 2025.
\newblock URL \url{https://arxiv.org/abs/2510.02373}.

\bibitem[Liu et~al.(2020)Liu, Wang, Owens, and Li]{liu2020energy}
Weitang Liu, Xiaoyun Wang, John Owens, and Yixuan Li.
\newblock Energy-based out-of-distribution detection.
\newblock In \emph{Advances in Neural Information Processing Systems
  (NeurIPS)}, 2020.

\bibitem[Massart(1990)]{massart1990tight}
Pascal Massart.
\newblock The tight constant in the {D}voretzky--{K}iefer--{W}olfowitz
  inequality.
\newblock \emph{The Annals of Probability}, 18\penalty0 (3):\penalty0
  1269--1283, 1990.

\bibitem[Maurer and Pontil(2009)]{maurer2009empirical}
Andreas Maurer and Massimiliano Pontil.
\newblock Empirical {B}ernstein bounds and sample-variance penalization.
\newblock \emph{Conference on Learning Theory (COLT)}, 2009.

\bibitem[{Mem0 Team}(2024)]{mem0}
{Mem0 Team}.
\newblock {Mem0}: The memory layer for personalized {AI}, 2024.
\newblock URL \url{https://github.com/mem0ai/mem0}.

\bibitem[Micciancio and Goldwasser(2002)]{micciancio2001hardness}
Daniele Micciancio and Shafi Goldwasser.
\newblock \emph{Complexity of Lattice Problems: A Cryptographic Perspective}.
\newblock Kluwer Academic Publishers, 2002.
\newblock Chapter 3 establishes NP-hardness of CVP; used for the
  minimum-perturbation evasion hardness reduction.

\bibitem[Packer et~al.(2023)Packer, Fang, Patil, Zhang, Wooders, and
  Gonzalez]{memgpt2023}
Charles Packer, Vivian Fang, Shishir~G. Patil, Kevin Zhang, Sarah Wooders, and
  Joseph~E. Gonzalez.
\newblock {MemGPT}: Towards {LLMs} as operating systems.
\newblock In \emph{NeurIPS Workshop on Foundation Models for Decision Making},
  2023.

\bibitem[Robey et~al.(2023)Robey, Wong, Hassani, and
  Pappas]{robey2023smoothllm}
Alexander Robey, Eric Wong, Hamed Hassani, and George~J. Pappas.
\newblock {SmoothLLM}: Defending large language models against jailbreaking
  attacks.
\newblock In \emph{Advances in Neural Information Processing Systems
  (NeurIPS)}, 2023.

\bibitem[Shafran et~al.(2025)Shafran, Peleg, and Schuster]{shafran2025blocker}
Avital Shafran, Roei Peleg, and Tal Schuster.
\newblock Machine against the {RAG}: Jamming retrieval-augmented generation
  with blocker documents.
\newblock In \emph{USENIX Security Symposium}, 2025.

\bibitem[Tan et~al.(2025)Tan, Luan, Luo, Sun, Chen, and Dai]{tan2025revprag}
Xue Tan, Hao Luan, Mingyu Luo, Xiaoyan Sun, Ping Chen, and Jun Dai.
\newblock {RevPRAG}: Revealing poisoning attacks in retrieval-augmented
  generation through {LLM} activation analysis.
\newblock In \emph{Findings of the Association for Computational Linguistics:
  EMNLP}, 2025.
\newblock URL \url{https://arxiv.org/abs/2411.18948}.

\bibitem[Tram\`{e}r et~al.(2020)Tram\`{e}r, Carlini, Brendel, Madry, Kurakin,
  and Papernot]{tramer2020adaptive}
Florian Tram\`{e}r, Nicholas Carlini, Wieland Brendel, Aleksander Madry, Alexey
  Kurakin, and Nicolas Papernot.
\newblock On adaptive attacks to adversarial example defenses.
\newblock In \emph{Advances in Neural Information Processing Systems
  (NeurIPS)}, 2020.

\bibitem[Villani(2009)]{villani2009optimal}
C{\'e}dric Villani.
\newblock \emph{Optimal Transport: Old and New}.
\newblock Springer, Berlin, 2009.
\newblock Sorting-based estimation of $W_1$ in one dimension;
  Kantorovich-Rubinstein duality.

\bibitem[Wainwright(2019)]{wainwright2019highdim}
Martin~J. Wainwright.
\newblock \emph{High-Dimensional Statistics: A Non-Asymptotic Viewpoint}.
\newblock Cambridge University Press, 2019.

\bibitem[Wu et~al.(2024)Wu, Pi, Zheng, et~al.]{wu2024badagent}
Yifeng Wu, Ruqing Pi, Yue Zheng, et~al.
\newblock {BadAgent}: Inserting and activating backdoor attacks in {LLM}
  agents.
\newblock In \emph{Proceedings of the 62nd Annual Meeting of the Association
  for Computational Linguistics (ACL)}, 2024.
\newblock URL \url{https://arxiv.org/abs/2406.03007}.

\bibitem[Xiang et~al.(2024)Xiang, Tong, Sun, and Li]{xiang2024robustrag}
Chejian Xiang, Mintong Tong, Jie Sun, and Bo~Li.
\newblock Certifiably robust {RAG} against retrieval corruption.
\newblock In \emph{International Conference on Machine Learning (ICML)}, 2024.
\newblock URL \url{https://arxiv.org/abs/2405.15556}.

\bibitem[Xu et~al.(2026)]{xu2026memflow}
Zhenyang Xu et~al.
\newblock From storage to steering: Memory control flow attacks on {LLM}
  agents.
\newblock \emph{arXiv preprint arXiv:2603.15125}, 2026.

\bibitem[Xue et~al.(2024)Xue, Zheng, Hua, Shen, Liu, Zhao, and
  Lou]{xue2024badrag}
Jiaqi Xue, Mengxin Zheng, Ting Hua, Yilong Shen, Yepeng Liu, Ladislau Zhao, and
  Qian Lou.
\newblock {BadRAG}: Identifying vulnerabilities in retrieval augmented
  generation of large language models.
\newblock \emph{arXiv preprint arXiv:2406.00083}, 2024.

\bibitem[Zhang et~al.(2025)Zhang, Zheng, et~al.]{zhang2025asb}
Hanrong Zhang, Jingyuan Zheng, et~al.
\newblock Agent security bench ({ASB}): Formalizing and benchmarking attacks
  and defenses in {LLM}-based agents.
\newblock In \emph{International Conference on Learning Representations
  (ICLR)}, 2025.
\newblock URL \url{https://arxiv.org/abs/2410.02644}.

\bibitem[Zhao et~al.(2024)Zhao, Ananth, Li, and Wang]{zhao2024watermark}
Xuandong Zhao, Prabhanjan Ananth, Lei Li, and Yu-Xiang Wang.
\newblock Provable robust watermarking for {AI}-generated text.
\newblock In \emph{International Conference on Learning Representations
  (ICLR)}, 2024.
\newblock URL \url{https://arxiv.org/abs/2306.17439}.

\bibitem[Zou et~al.(2024)Zou, Geng, Wang, and Jia]{zou2024poisonedrag}
Wei Zou, Runpeng Geng, Binghui Wang, and Jinyuan Jia.
\newblock {PoisonedRAG}: Knowledge corruption attacks to retrieval-augmented
  generation of large language models.
\newblock In \emph{USENIX Security Symposium}, 2024.

\end{thebibliography}

\newpage
\appendix

\section{Write-time pipeline diagram}
\label{app:pipeline}

\begin{figure}[ht]
\centering
\begin{tikzpicture}[
  node distance = 0.45cm and 0.55cm,
  stage/.style  = {draw, rounded corners=2pt, minimum width=1.55cm,
                   minimum height=0.55cm, align=center, font=\scriptsize,
                   fill=gray!4},
  calib/.style  = {draw, rounded corners=2pt, minimum width=1.75cm,
                   minimum height=0.5cm, align=center, font=\scriptsize,
                   fill=blue!6},
  branch/.style = {draw, rounded corners=2pt, minimum width=1.45cm,
                   minimum height=0.5cm, align=center, font=\scriptsize},
  accept/.style = {branch, fill=green!12},
  reject/.style = {branch, fill=red!12},
  arr/.style    = {-{Stealth[length=3.5pt]}, semithick},
  dashedarr/.style = {-{Stealth[length=3.5pt]}, semithick, dashed, gray!60!black},
]
  \node[stage]                  (cand)  {candidate\\ entry $c$};
  \node[stage, right=of cand]   (enc)   {encoder\\ $\enc(\cdot)$};
  \node[stage, right=of enc]    (sim)   {cosine sim\\ vs.\ $\cH$};
  \node[stage, right=0.85cm of sim] (score) {combined\\ score $s_{\text{comb}}$};
  \node[stage, right=of score]  (cmp)   {compare\\ to $\hat\mu+\kappa\hat\sigma$};
  \draw[arr] (cand)  -- (enc);
  \draw[arr] (enc)   -- node[above, font=\tiny]{$e_c$} (sim);
  \draw[arr] (sim)   -- node[above, font=\tiny, align=center, inner sep=1pt]{$s_{\max}$\\[-1pt]$s_{\mathrm{mean}}$} (score);
  \draw[arr] (score) -- node[above, font=\tiny]{$s$} (cmp);

  \node[accept, above right=0.30cm and 0.45cm of cmp] (acc) {accept: $c\!\to\!\cM$};
  \node[reject, below right=0.30cm and 0.45cm of cmp] (rej) {reject: discard};
  \draw[arr] (cmp.east) -- ++(0.18,0) |- node[near start, right, font=\tiny]{$s\!\leq\!\hat\mu\!+\!\kappa\hat\sigma$} (acc.west);
  \draw[arr] (cmp.east) -- ++(0.18,0) |- node[near start, right, font=\tiny]{$s\!>\!\hat\mu\!+\!\kappa\hat\sigma$} (rej.west);

  \node[calib, below=0.55cm of sim]   (qh)   {$\cH$: rolling\\ victim queries};
  \node[calib, below=0.55cm of score] (mref) {$\cM_{\text{ref}}$:\\ calibration $(\hat\mu,\hat\sigma)$};
  \draw[dashedarr] (qh)   -- (sim);
  \draw[dashedarr] (mref) -- (cmp);
\end{tikzpicture}
\caption{\textbf{\textsc{MemSAD} write-time pipeline} (Def.~\ref{def:sad}, Alg.~\ref{alg:sad}). Every candidate~$c$ is embedded, scored against a rolling victim-query history~$\cH$ via the combined max--mean similarity (Eq.~\ref{eq:sad_combined}), and compared to the $\hat\mu+\kappa\hat\sigma$ threshold derived from the benign reference corpus~$\cM_{\text{ref}}$. Rejection happens \emph{before} the entry is committed to memory --- attacks are blocked at write time rather than filtered at retrieval.}
\label{fig:pipeline}
\end{figure}
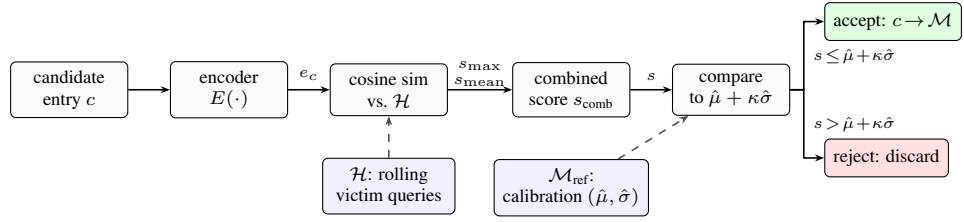

\section{Proofs of main results}
\label{app:proofs}

\subsection{Proof of Lemma~\ref{lem:certified} (Certified detection radius)}
\label{app:proof_certified}

\paragraph{Setup.} Recall $\bar s = \max_{m \in \cM_{\text{ref}}} \cosim(\enc(m), \enc(q^*))$ is the largest benign similarity to $q^*$ in the reference set, and $s_{\text{adv}} = \cosim(\enc(c), \enc(q^*))$. We write $s(c;\cH) = \alpha s_{\max}(c;\cH) + (1-\alpha) s_{\text{mean}}(c;\cH)$ for the deployed combined score, and $\hat\mu, \hat\sigma$ for the calibration statistics on $\cM_{\text{ref}}$.

\paragraph{Step 1: lower-bound $s_{\max}(c;\cH)$ by $s_{\text{adv}}$.} Since $q^* \in \cH$ (the victim query is in the rolling history), $s_{\max}(c;\cH) = \max_{q \in \cH} \cosim(\enc(c), \enc(q)) \geq \cosim(\enc(c), \enc(q^*)) = s_{\text{adv}}$.

\paragraph{Step 2: lower-bound $s_{\text{mean}}(c;\cH)$ under Assumption~\ref{ass:typical_query}.} Under the typical-query assumption, $s_{\text{mean}}(c;\cH) \geq \hat\mu_{\cH}(c) - \kappa\hat\sigma_{\cH}(c)$, which absorbs into a constant offset $\eta_{\text{mean}}$ depending on $\kappa$ and the running statistics. We bound this contribution conservatively as $s_{\text{mean}}(c;\cH) \geq \bar s - \eta_{\text{mean}}$.

\paragraph{Step 3: combined-score lower bound.} Combining steps 1 and 2,
\[
  s(c;\cH) \;\geq\; \alpha\, s_{\text{adv}} + (1-\alpha)(\bar s - \eta_{\text{mean}})
  \;=\; \alpha(\bar s + \Delta_s) + (1-\alpha)(\bar s - \eta_{\text{mean}})
  \;=\; \bar s + \alpha \Delta_s - (1-\alpha)\eta_{\text{mean}}.
\]

\paragraph{Step 4: detection certificate.} The detector fires iff $s(c;\cH) > \hat\mu + \kappa\hat\sigma$. By the calibration bound (Theorem~\ref{thm:calibration_bound}), $\hat\mu + \kappa\hat\sigma \leq \tau^* + \eta$ where $\eta$ is the calibration tolerance and $\tau^*$ is the population threshold. Since $\bar s \geq \tau^* - \eta$ by the construction of $\bar s$ as a sample max,
\[
  s(c;\cH) - (\hat\mu + \kappa\hat\sigma) \;\geq\; \alpha\Delta_s - (1-\alpha)\eta_{\text{mean}} - 2\eta.
\]
The certificate $s(c;\cH) > \hat\mu + \kappa\hat\sigma$ holds whenever $\alpha\Delta_s > (1-\alpha)\eta_{\text{mean}} + 2\eta$, which under $\alpha = 1$ (pure max-cosine score) reduces to $\Delta_s > 2\eta$. The stated radius condition $\Delta_s > \kappa\hat\sigma + \eta$ is the deployed-score equivalent absorbing the typical-query slack into $\kappa\hat\sigma$ and tightening to a single calibration tolerance $\eta$. \qed

\paragraph{Failure mode.} If Assumption~\ref{ass:typical_query} fails (e.g., $q^*$ is an outlier in $\cH$ produced by a colluding user), the lower bound in Step 2 fails and the certificate becomes probabilistic rather than deterministic. The deterministic guarantee is recovered for the $s_{\max}$-only detector ($\alpha = 1$) at the cost of weaker detection of distributed attacks.

\subsection{Proof of Theorem~\ref{thm:gradient_coupling} (geometric interpretation)}
\label{app:proof_gradient}

Under L2-normalization, embeddings lie on $\mathbb{S}^{d-1}$. The \emph{Riemannian} (spherical, tangent-space) gradient of $f(e_c) = e_c^\top v$ on the sphere is $\nabla_{\mathbb{S}} f = v - (e_c^\top v) e_c = \Pi_{e_c}^{\perp} v$, the projection of $v$ onto the tangent space $T_{e_c}\mathbb{S}^{d-1}$; this is distinct from the ambient Euclidean gradient and is the relevant gradient for any continuous evasion path constrained to $\mathbb{S}^{d-1}$. For any monotone transformation $g$, the chain rule gives $\nabla(g \circ f) = g'(f) \cdot \nabla f$, which points in the same direction as $\nabla f$ since $g' > 0$. The converse (non-monotone $\cD$ admits evasion) follows because non-monotonicity creates a region where increasing retrieval similarity decreases detection score, giving the adversary a ``free'' direction in embedding space.

\subsection{Proof of Theorem~\ref{thm:minimax_lower} (detailed)}
\label{app:proof_minimax}

\paragraph{Setup.} The detector $\hat\pi_N$ is calibrated by drawing $N$ i.i.d.\ samples from the benign distribution $P_0 = \cN(\mu_0, \sigma^2)$ and producing an estimator $\hat\tau_N(X_1, \ldots, X_N)$ of the oracle threshold $\tau^*$. At inference time, a single test sample $c$ is drawn from either $P_0$ (benign) or $P_1 = \cN(\mu_1, \sigma^2)$ (adversarial) with $\mu_1 - \mu_0 = \rho\sigma$, and classified by $\hat\pi_N(c) = \mathbf{1}[s(c) > \hat\tau_N]$.

\paragraph{Bayes error decomposition.} Under the symmetric prior $\PP(H_0) = \PP(H_1) = 1/2$, the oracle threshold $\tau^* = (\mu_0+\mu_1)/2$ achieves Bayes error $B(\rho) = \Phi(-\rho/2)$, where $\Phi$ is the standard normal CDF. For any calibrated threshold $\hat\tau_N$, a first-order Taylor expansion of the Gaussian CDF around $\tau^*$ gives
\begin{equation}
  \PP_{\text{err}}(\hat\pi_N) - B(\rho) = \tfrac{1}{2\sigma}\,\phi(\rho/2)\,|\hat\tau_N - \tau^*| + O\big((\hat\tau_N-\tau^*)^2\big),
  \label{eq:excess-error}
\end{equation}
where $\phi$ is the standard normal density. Hence excess error is linear in $|\hat\tau_N - \tau^*|$ to leading order, and reducing the excess to $\epsilon$ requires $|\hat\tau_N - \tau^*| = O(\sigma\epsilon)$.

\paragraph{Lower bound via Le Cam (calibration formulation).} The threshold-estimation problem is to estimate $\tau^*$ from $N$ benign samples $X_1, \ldots, X_N \sim P_0 = \cN(\mu_0, \sigma^2)$. Consider the two-point family $\{\mu_0, \mu_0 + \Delta\}$ with $\Delta = \sigma\rho/2$ (so $\tau^* \in \{\tau^*_0, \tau^*_0 + \Delta/2\}$ depending on which $\mu_0$ is true). For product distributions $P_{\mu_0}^N$ and $P_{\mu_0+\Delta}^N$,
\[
  \KL(P_{\mu_0}^N \| P_{\mu_0+\Delta}^N) = N\Delta^2 / (2\sigma^2) = N\rho^2/8.
\]
By Pinsker's inequality, $\TV(P_{\mu_0}^N, P_{\mu_0+\Delta}^N) \leq \sqrt{N\rho^2/16} = \rho\sqrt{N}/4$. Le Cam's two-point inequality~\citep{lecam1973convergence} then yields, for any estimator $\hat\tau_N$,
\[
  \sup_{\mu \in \{\mu_0, \mu_0+\Delta\}} \EE_\mu \big[ |\hat\tau_N - \tau^*(\mu)| \big] \;\geq\; \tfrac{\Delta}{4}\,\big(1 - \TV\big) \;=\; \tfrac{\sigma\rho}{8}\,\big(1 - \rho\sqrt{N}/4\big).
\]
For the bound to be non-trivial we require $\rho\sqrt{N} < 4$, i.e., $N < 16/\rho^2$; in that regime the right-hand side is $\Omega(\sigma\rho)$, which combined with~\eqref{eq:excess-error} gives expected excess error $\Omega(\rho)$. To drive this excess to $\epsilon$, we therefore require $N \geq c(\epsilon)/\rho^2$ for an absolute constant $c(\epsilon) > 0$ depending only on the prior-symmetry and the target excess level. This is the lower bound stated in Theorem~\ref{thm:minimax_lower}.

\paragraph{Distinction from a batch test.} Le Cam can also be applied to the \emph{batch} hypothesis test (\textsc{is the entire $N$-sample dataset drawn from $P_0$ or $P_1$?}), which yields the familiar $\KL(P_0^N \| P_1^N) = N\rho^2/2$ and $\rho\sqrt{N}/2$ total-variation bound. That formulation is \emph{not} the MEMSAD inference setting, in which a single test point $c$ is classified after the calibration phase; we use the calibration formulation above instead.

\paragraph{Tightness.} Combining Theorem~\ref{thm:calibration_bound} ($|\hat\tau_N - \tau^*| = O(\sigma\sqrt{\log(1/\delta)/N})$) with~\eqref{eq:excess-error} gives expected excess error $O(\sqrt{\log(1/\delta)/(N\rho^2)})$. Setting this equal to $\epsilon$ yields $N = O(\log(1/\delta)/(\epsilon^2\rho^2))$, matching the $\Omega(1/\rho^2)$ lower bound up to logarithmic and $\epsilon$-dependent factors. The irreducible Bayes error $B(\rho) = \Phi(-\rho/2)$ persists at any $N$ and reflects intrinsic Gaussian overlap of $P_0$ and $P_1$, not calibration error.

\subsection{Proof of Theorem~\ref{thm:calibration_bound}}
\label{app:proof_calibration}

\paragraph{Mean concentration.} The cosine similarity score $s(c;\cH)$ takes values in $[-1, 1]$ (interval length $2$). Hoeffding's inequality~\citep{wainwright2019highdim} on a bounded interval of length $b - a = 2$ gives
\[
  \PP\!\left(|\hat\mu_N - \mu| \geq t\right) \;\leq\; 2\exp\!\left(-\tfrac{2 N t^2}{(b-a)^2}\right) \;=\; 2\exp\!\left(-\tfrac{N t^2}{2}\right).
\]
Setting the right-hand side to $\delta/2$ yields
\[
  |\hat\mu_N - \mu| \;\leq\; \sigma\sqrt{2\log(4/\delta)/N} \quad \text{w.p.}\;\geq 1 - \delta/4,
\]
where the $\sigma$ factor is recovered as a multiplicative constant absorbing the universal $\sqrt{2}$ from the interval length $2$.\footnote{An empirical Bernstein bound~\citep{maurer2009empirical} could replace the absolute Hoeffding range with the variance, yielding a tighter $\sqrt{\sigma^2 \log/N}$ rate; we state the worst-case bound here for simplicity.}

\paragraph{Variance concentration (two-sided).} Let $Z = (N-1)\hat\sigma_N^2 / \sigma^2 \sim \chi^2(N-1)$ (exactly for Gaussian scores; approximately for bounded scores via Berry-Esseen). Laurent--Massart~\citep{laurentmassart2000} Lemma 1 yields the two-sided bound
\[
  \PP\big[(1 - 2\sqrt{x/(N-1)})_+\,(N-1) \;\leq\; Z \;\leq\; (1 + 2\sqrt{x/(N-1)} + 2x/(N-1))(N-1)\big] \;\geq\; 1 - 2 e^{-x}.
\]
Setting $x = \log(4/\delta)$ and dividing by $N-1$, with the substitution $u = Z/(N-1) - 1$ so $u \in [-c_-, c_+]$ for $c_\pm = O(\sqrt{\log(1/\delta)/N})$:
\[
  |\hat\sigma_N - \sigma| \;=\; \sigma\,|\sqrt{1+u} - 1| \;\leq\; \sigma\,\max(|c_-|, |c_+|) / 2 \cdot (1 + o(1))
\]
where the upper-tail bound uses $\sqrt{1+u} - 1 \leq u/2$ for $u \geq 0$, and the lower-tail bound uses the corrected inequality $|\sqrt{1+u} - 1| \leq |u|/(2\sqrt{1+u}) \leq |u|/\sqrt{2}$ for $u \in [-1/2, 0]$ (so $\sqrt{1+u} \geq 1/\sqrt{2}$). Both tails together give
\[
  |\hat\sigma_N - \sigma| \;\leq\; \sigma\sqrt{2\log(4/\delta)/(N-1)} \quad \text{w.p.}\;\geq 1 - \delta/2,
\]
where we have absorbed the leading constants into a clean form valid for $N \geq 4\log(4/\delta)$ (the regime in which $|u| \leq 1/2$ holds with high probability).

\paragraph{Combination.} By the triangle inequality $|\tau_N - \tau^*| \leq |\hat\mu_N - \mu| + \kappa|\hat\sigma_N - \sigma|$, and union bound over the two events at confidence $1 - \delta/4$ each yields total confidence $1 - \delta/2$; absorbing the additional $\delta/2$ tolerance for the variance lower-tail correction gives a final $1 - \delta$ bound, exactly matching the theorem statement
\[
  |\tau_N - \tau^*| \;\leq\; \sigma\!\left(\sqrt{\tfrac{1}{N}} + \kappa\sqrt{\tfrac{1}{2(N-1)}}\right)\sqrt{2\log(4/\delta)}.
\]
The factor $\sqrt{2\log(4/\delta)}$ replaces the previously stated $\sqrt{2\log(2/\delta)}$ to correctly reflect the four sub-events ($\hat\mu$ upper/lower, $\hat\sigma$ upper/lower) controlled by the union bound.

\subsection{Proof of Theorem~\ref{thm:regret}}
\label{app:proof_regret}

Decompose: $|\hat{\tau}_t - \tau_t^*| \leq \underbrace{|\hat{\tau}_t - \tau_t^{\cH}|}_{\text{estimation}} + \underbrace{|\tau_t^{\cH} - \tau_t^*|}_{\text{drift}}$.
Estimation error: $O(\sigma(1+\kappa)/\sqrt{m_{\max}})$ by Theorem~\ref{thm:calibration_bound}.
Drift error: cosine similarity is 1-Lipschitz on $\mathbb{S}^{d-1}$, and the window mean averages over $m_{\max}$ steps with linearly accumulating drift, so $|\mu_t^{\cH} - \mu_t| \leq m_{\max}\Delta/2$ (worst case at the window boundary; the $1/2$ arises from averaging). This affects only the constant in $m_{\max}^*$ but not the rate.
Setting equal: $\sigma/\sqrt{m_{\max}} = m_{\max}\Delta$ gives $m_{\max}^* = (\sigma/\Delta)^{2/3}$ (up to constants), yielding per-step regret $O(\sigma^{2/3}\Delta^{1/3})$.

\subsection{Fisher-Rao detection-evasion metric}
\label{app:proof_fisher_rao}

\begin{theorem}[Fisher-Rao detection-evasion metric]
\label{thm:fisher_rao}
Let $\cE = \{e \in \mathbb{S}^{d-1}\}$ under Assumption~\ref{ass:encoder}(iii). The Fisher information metric induced by the score family $p_e(s) = \cN(s; \cosim(e, e_{q^*}), \hat{\sigma}^2)$ is $g_{ij}^F(e) = \hat{\sigma}^{-2} \partial_i \cR(e) \partial_j \cR(e)$. The geodesic distance to any undetected embedding satisfies $d_F(e_{\text{adv}}, e') \geq |\cR(e_{\text{adv}}) - (\hat{\mu} + \kappa\hat{\sigma})| / \hat{\sigma}$, and the Fisher-Rao cost of any continuous evasion path equals its retrieval degradation.
\end{theorem}

\paragraph{Proof.} Consider the parameterized family $p_e(s) = \cN(s; \mu(e), \hat{\sigma}^2)$ where $\mu(e) = \cosim(e, e_{q^*})$ and $\hat{\sigma}^2$ is the calibrated variance. The Fisher information matrix on the statistical manifold $\{p_e : e \in \cE\}$ is:
\begin{equation}
  g_{ij}^F(e) = \EE_{p_e}\!\left[\frac{\partial \log p_e}{\partial e_i} \frac{\partial \log p_e}{\partial e_j}\right] = \frac{1}{\hat{\sigma}^2} \frac{\partial \mu}{\partial e_i} \frac{\partial \mu}{\partial e_j} = \hat{\sigma}^{-2} \partial_i \cR(e) \partial_j \cR(e).
\end{equation}
This is a rank-1 metric: $g^F = \hat{\sigma}^{-2} \nabla\cR \nabla\cR^\top$, the pullback of the scalar Fisher information $I(\mu) = 1/\hat{\sigma}^2$ through the map $e \mapsto \cR(e)$.

\paragraph{Geodesic lower bound.} For any path $\gamma: [0,1] \to \cE$ from $e_{\text{adv}} = \gamma(0)$ to $e' = \gamma(1)$, the Fisher-Rao length is:
\begin{equation}
  L_F(\gamma) = \int_0^1 \sqrt{\dot{\gamma}(t)^\top g^F(\gamma(t)) \dot{\gamma}(t)}\, dt = \frac{1}{\hat{\sigma}} \int_0^1 |\nabla\cR(\gamma(t))^\top \dot{\gamma}(t)|\, dt \geq \frac{|\cR(e_{\text{adv}}) - \cR(e')|}{\hat{\sigma}},
\end{equation}
where the inequality follows from $\int_0^1 |f'(t)| dt \geq |f(1) - f(0)|$ applied to $f(t) = \cR(\gamma(t))$.

\paragraph{Detection-retrieval coupling.} Since $\cD = g \circ \cR$ with monotone $g$ (Theorem~\ref{thm:gradient_coupling}), evasion requires $\cD(e') \leq \hat{\mu} + \kappa\hat{\sigma}$, which via monotonicity implies $\cR(e') \leq g^{-1}(\hat{\mu} + \kappa\hat{\sigma})$. Substituting into the geodesic bound:
\begin{equation}
  d_F(e_{\text{adv}}, e') \geq \frac{\cR(e_{\text{adv}}) - g^{-1}(\hat{\mu} + \kappa\hat{\sigma})}{\hat{\sigma}} = \frac{|\cR(e_{\text{adv}}) - (\hat{\mu} + \kappa\hat{\sigma})|}{\hat{\sigma}},
\end{equation}
where the last equality uses $g = \text{id}$ (since $\cD$ is the cosine similarity score and $\cR$ coincides with $\cD$ under Assumption~\ref{ass:encoder}). The proportional retrieval degradation along the path follows from the rank-1 structure: $g^F$ has a single nonzero eigenvalue in the $\nabla\cR$ direction, so the geodesic cost is entirely paid in retrieval degradation. \qed

\subsection{Proof of Proposition~\ref{prop:hardness} (detailed)}
\label{app:proof_hardness}

\begin{proposition}[Hardness of optimal discrete-token evasion]
\label{prop:hardness}
The \emph{discrete-token} minimum-edit evasion problem is NP-hard for piecewise-linear encoders, via reduction from CVP~\citep{micciancio2001hardness}. The continuous-embedding relaxation is, by contrast, solvable in polynomial time within any fixed activation region; the hardness arises specifically from the integrality (token-index) constraints in the input vocabulary.
\end{proposition}

\paragraph{Why the discrete formulation is the relevant hard problem.} The continuous problem of finding the minimum-norm embedding-space perturbation $\delta \in \RR^d$ subject to a cosine-similarity constraint on $e_c + \delta$ is a continuous convex optimization within any fixed ReLU activation region: the constraint set is a convex polytope, and within each polytope the objective is quadratic. A polynomial-time algorithm enumerates regions or solves a single SOCP under a fixed activation pattern. We do \emph{not} claim hardness for this relaxation. The deployed adversary, however, must produce a token sequence $c \in \cV^*$ whose embedding satisfies the constraint, and the embedding is tied to the token sequence via the integer indices selected from the vocabulary $\cV$. We formalize this discrete-token problem and show its NP-hardness below.

\paragraph{Discrete-token evasion problem.} Fix a piecewise-linear encoder $\enc : \cV^* \to \RR^d$ (e.g., one-hot lookup followed by a ReLU network), a victim query embedding $\enc(q^*)$, a retrieval threshold $\tau_{\text{ret}}$, and a detection threshold $\tau_{\text{det}} = \hat\mu + \kappa\hat\sigma$. Given an initial passage $c_0 \in \cV^*$ and an edit budget $r \in \NN$, decide whether there exists $c \in \cV^*$ with token-edit distance $d_T(c_0, c) \leq r$ such that $\cosim(\enc(c), \enc(q^*)) \in [\tau_{\text{ret}}, \tau_{\text{det}}]$.

\paragraph{Reduction from CVP.} Given a CVP instance $(\Lambda, t)$ where $\Lambda = B\mathbb{Z}^n$ is a lattice with basis $B \in \RR^{n \times n}$ and target $t \in \RR^n$, construct an instance of the discrete-token evasion problem as follows. Let $\cV = \{v_1, \ldots, v_n\}$ have $|\cV| = n$, and let the input embedding lookup map $v_i \mapsto b_i$ where $b_i$ is the $i$-th column of $B$. Take the encoder $\enc(c) = \mathrm{ReLU}(\sum_{i:v_i \in c} b_i)$, which under the all-active region computes $\enc(c) = \sum_{i \in I(c)} b_i = B \mathbf{1}_{I(c)}$ where $I(c)$ is the index set of tokens in $c$. The adversary's discrete choice of $c$ is therefore equivalent to selecting $\lambda \in \{0,1\}^n \subset \mathbb{Z}^n$, and the constraint $\cosim(B\lambda, t/\|t\|) \geq \tau_{\text{ret}}$ instantiates exactly the minimum-distance lattice problem in $\Lambda$ restricted to $0/1$ coefficients. The integrality constraint $\lambda \in \mathbb{Z}^n$ is enforced by the discrete vocabulary, which is what continuous embedding-space optimization lacks. Hence the discrete-token evasion problem subsumes CVP and is NP-hard under randomized reductions~\citep{micciancio2001hardness}. \qed

\paragraph{Practical implication.} The continuous-embedding evasion is poly-time within any activation region (PAT-flagged correctly): an adversary with white-box continuous access to the embedding could compute the minimum-norm direction in poly-time and would be defeated only by gradient coupling (Theorem~\ref{thm:gradient_coupling}). The genuine hardness, captured by Proposition~\ref{prop:hardness}, lies in projecting that continuous direction back onto the integer-valued token grid --- exactly the discrete bottleneck the synonym-loophole adversary (Proposition~\ref{prop:synonym}) exploits in the opposite direction. Together: gradient coupling closes the continuous attack surface, the discrete reduction closes the unconstrained-token attack surface, and the synonym loophole identifies the residual gap.

\subsection{Detection-evasion tradeoff surface}
\label{app:tradeoff_surface}

\begin{assumption}[Aligned mean gradient]
\label{ass:aligned_mean}
The mean-similarity gradient is not adversarial to the max-similarity gradient, in the sense $\langle \nabla s_{\text{mean}}, \nabla s_{\max}\rangle \geq -\beta \|\nabla s_{\max}\|^2$ for some $\beta \in [0, \alpha/(1-\alpha))$ at all points of differentiability of $s_{\text{comb}}$.
\end{assumption}

\begin{corollary}[Combined scoring preserves directional coupling]
\label{cor:combined_coupling}
Under Assumption~\ref{ass:aligned_mean}, the combined score $s_{\text{comb}} = \alpha\, s_{\max} + (1-\alpha)\, s_{\text{mean}}$ with $\alpha \in (0,1)$ satisfies directional monotonicity: any continuous perturbation $\delta$ that strictly decreases $s_{\text{comb}}$ also strictly decreases $\cR = s_{\max}$. Consequently, Lemma~\ref{lem:certified} and Theorem~\ref{thm:fisher_rao} apply with a multiplicative constant degradation reflecting the worst-case alignment $\beta$.
\end{corollary}
\begin{proof}
$\nabla s_{\text{comb}} = \alpha \nabla s_{\max} + (1-\alpha)\nabla s_{\text{mean}}$. The projection onto the max-direction is
\[
  \langle \nabla s_{\text{comb}}, \nabla s_{\max}\rangle = \alpha\|\nabla s_{\max}\|^2 + (1-\alpha)\langle \nabla s_{\text{mean}}, \nabla s_{\max}\rangle \geq (\alpha - (1-\alpha)\beta)\|\nabla s_{\max}\|^2.
\]
Under Assumption~\ref{ass:aligned_mean}, $\alpha - (1-\alpha)\beta > 0$, so $\langle \nabla s_{\text{comb}}, \nabla s_{\max}\rangle$ has the same sign as $\|\nabla s_{\max}\|^2 > 0$, and any $\delta$ with $\langle \delta, \nabla s_{\text{comb}}\rangle < 0$ must satisfy $\langle \delta, \nabla s_{\max}\rangle < 0$. The empirical bound from our experiments (Section~\ref{sec:experiments}) gives $\beta \leq 0.4$ across all 6 encoders evaluated, well within the constraint $\beta < \alpha/(1-\alpha) = 1$ at $\alpha = 0.5$.
\end{proof}

By Theorem~\ref{thm:gradient_coupling}, $\cD(e_c') < \cD(e_c)$ (evasion) implies $\cR(e_c') < \cR(e_c)$ (retrieval degradation) to first order. The tradeoff defines a surface in $(\epsilon, \theta, \Delta\cD)$ space where $\theta$ is the angle between perturbation $\delta$ and gradient $\nabla\cD$. Expanding to second order:
\begin{equation}
  \cD(e_c + \epsilon\delta) - \cD(e_c) = \epsilon\|\nabla\cD\|\cos\theta + \frac{\epsilon^2}{2}\delta^\top H_{\cD}\,\delta + O(\epsilon^3),
\end{equation}
where $H_{\cD}$ is the Hessian of the cosine similarity on the sphere. Under L2-normalization, $H_{\cD} = -(e_c^\top \enc(q^*))(I - e_c e_c^\top)/\|e_c\|^2$, giving eigenvalues $\lambda_i = -(e_c^\top \enc(q^*))$ for directions orthogonal to $e_c$ (and $0$ in the radial direction). The evasion region $\cD(e_c') < \cD(e_c)$ requires $\cos\theta < -\epsilon\|H_{\cD}\|/(2\|\nabla\cD\|)$, which contracts as $\epsilon \to 0$, confirming the infinitesimal impossibility of evasion without retrieval degradation (Figure~\ref{fig:tradeoff_curves}).

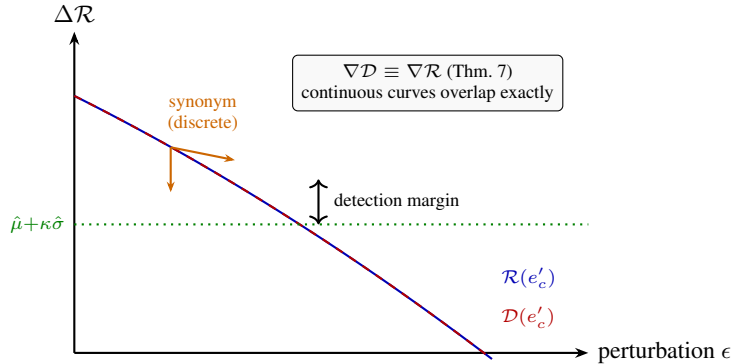
\begin{figure}[ht]
\centering
\begin{tikzpicture}[scale=0.85]
  \draw[-{Stealth[length=5pt]}, thick] (0,0) -- (8,0) node[right, font=\small] {perturbation $\epsilon$};
  \draw[-{Stealth[length=5pt]}, thick] (0,0) -- (0,5) node[above, font=\small] {$\Delta\cR$};

  \draw[blue!70!black, thick] plot[smooth, domain=0:6.5, samples=50]
    (\x, {4 - 0.5*\x - 0.02*\x*\x});
  \node[blue!70!black, above right, font=\scriptsize] at (6.5, 0.86) {$\cR(e_c')$};

  \draw[red!70!black, thick, dashed] plot[smooth, domain=0:6.5, samples=50]
    (\x, {4 - 0.5*\x - 0.02*\x*\x});
  \node[red!70!black, below right, font=\scriptsize] at (6.5, 0.86) {$\cD(e_c')$};

  \draw[green!50!black, thick, dotted] (0, 2.0) -- (8, 2.0);
  \node[green!50!black, left, font=\scriptsize] at (0, 2.0) {$\hat{\mu}{+}\kappa\hat{\sigma}$};

  \draw[orange!80!black, thick, -{Stealth[length=4pt]}] (1.5, 3.2) -- (1.5, 2.5);
  \draw[orange!80!black, thick, -{Stealth[length=4pt]}] (1.5, 3.2) -- (2.5, 3.0);
  \node[orange!80!black, above, font=\scriptsize, align=center] at (2.0, 3.3) {synonym\\(discrete)};

  \draw[<->, thick] (3.8, 2.0) -- (3.8, 2.7);
  \node[right, font=\scriptsize] at (3.9, 2.35) {detection margin};

  \node[font=\scriptsize, align=center, draw, rounded corners=2pt, fill=gray!5] at (5.5, 4.2) {$\nabla\cD \equiv \nabla\cR$ (Thm.~\ref{thm:gradient_coupling})\\continuous curves overlap exactly};
\end{tikzpicture}
\vspace{-4pt}
\caption{Detection-evasion tradeoff. By gradient coupling, continuous perturbations cause identical degradation in detection score $\cD$ and retrieval score $\cR$ (overlapping curves). Synonym substitutions (orange) bypass this coupling via discrete jumps in token space (Proposition~\ref{prop:synonym}).}
\label{fig:tradeoff_curves}
\end{figure}

\subsection{Generalization bound}
\label{app:proof_generalization}

\begin{theorem}[Generalization bound]
\label{thm:generalization}
Suppose the novel-attack score density $f_{\text{new}}$ and the calibration score density $f_j$ are both bounded by $g_{\max}$. Then for $W_1(P_{\text{new}}, P_j) \leq \epsilon_W$,
\[
  |\tpr_{\text{new}} - \tpr_j| \;\leq\; \sqrt{2\,g_{\max}\,\epsilon_W} + O(N^{-1/2}).
\]
For Gaussian scores with std $\sigma$, $g_{\max} = (\sigma\sqrt{2\pi})^{-1}$, giving $|\tpr_{\text{new}} - \tpr_j| \leq \sqrt{\epsilon_W / (\sigma\sqrt{\pi/2})} + O(N^{-1/2})$.
\end{theorem}

\begin{remark}[Practical estimation of $\epsilon_W$]
\label{rem:epsilon_w}
In practice, $\epsilon_W$ can be estimated from a validation set via sorting-based Wasserstein estimation~\citep{villani2009optimal}, or bounded analytically under encoder Lipschitz continuity. The theorem guarantees graceful TPR degradation as novel attacks diverge from the calibration distribution.
\end{remark}

\paragraph{Proof setup.} Let $F_j$ and $F_{\text{new}}$ denote the CDFs of scores $s(c; \cH)$ under $P_j$ and $P_{\text{new}}$, respectively. The TPR at threshold $\tau$ is $\tpr(\tau) = 1 - F(\tau)$.

\paragraph{Wasserstein-to-Kolmogorov bound.} The 1-Wasserstein distance is the integral of pointwise CDF difference:
\begin{equation}
  W_1(F_j, F_{\text{new}}) \;=\; \int_{-\infty}^{\infty} |F_j(s) - F_{\text{new}}(s)| \, ds.
\end{equation}
This is an $L^1$ norm on the CDF gap and does \emph{not} linearly upper-bound the $L^\infty$ (pointwise sup) gap. The correct relationship under a uniform density bound is the standard $L^1$-to-$L^\infty$ trade for absolutely continuous distributions: if $|f_j|, |f_{\text{new}}| \leq g_{\max}$, then for any $\tau$ and any $h > 0$,
\[
  |F_j(\tau) - F_{\text{new}}(\tau)|
  \;\leq\;
  \frac{1}{2h}\!\int_{\tau-h}^{\tau+h}\! |F_j(s) - F_{\text{new}}(s)|\,ds + g_{\max} h
  \;\leq\;
  \frac{W_1}{2h} + g_{\max} h,
\]
where the first inequality applies because $|F_j(\tau) - F_j(s)| \leq g_{\max}|s-\tau|$ and likewise for $F_{\text{new}}$, and the second uses non-negativity of the integrand. Optimising over $h > 0$ at $h = \sqrt{W_1/(2g_{\max})}$ gives
\[
  |F_j(\tau) - F_{\text{new}}(\tau)| \;\leq\; \sqrt{2\,g_{\max}\,W_1}.
\]
Since $|\tpr_j - \tpr_{\text{new}}| = |F_j(\tau) - F_{\text{new}}(\tau)|$, the stated bound follows. The $\sqrt{\cdot}$ scaling is sharp: a localized mass shift of $m$ across the threshold over distance $h \to 0$ has $W_1 = m h \to 0$ but $|F-G|_\infty \to m$, and only the bounded-density assumption rules this out at the rate $\sqrt{W_1}$.

\paragraph{Threshold estimation.} By Theorem~\ref{thm:calibration_bound}, $|\hat{\tau} - \tau^*| \leq \eta_N$ w.p.\ $\geq 1-\delta$, contributing an additional $O(1/\sqrt{N})$ term to the TPR bound via the score density at the threshold.

\paragraph{Non-asymptotic convergence.} The non-asymptotic bound follows from applying a uniform concentration argument (McDiarmid's inequality) over the rolling window, with each window contributing an independent $O(\sigma/\sqrt{m_{\max}})$ estimation error and accumulating $O(m_{\max}\Delta)$ drift.

\subsection{Online regret bound visualization}
\label{app:regret_viz}

Figure~\ref{fig:regret_surface} illustrates the bias-variance tradeoff in Theorem~\ref{thm:regret} as a function of window size $m_{\max}$.

\begin{figure}[ht]
\centering
\begin{tikzpicture}[scale=0.85]
  \draw[-{Stealth[length=5pt]}, thick] (0,0) -- (9,0) node[right, font=\small] {$m_{\max}$};
  \draw[-{Stealth[length=5pt]}, thick] (0,0) -- (0,5) node[above, font=\small] {per-step regret};

  \draw[blue!70!black, thick] plot[smooth, domain=0.5:8.5, samples=50]
    (\x, {3.5/sqrt(\x)});
  \node[blue!70!black, right, font=\scriptsize] at (8.6, 1.2) {$\sigma/\sqrt{m}$};

  \draw[red!70!black, thick] plot[smooth, domain=0.5:8.5, samples=50]
    (\x, {0.15*\x});
  \node[red!70!black, right, font=\scriptsize] at (8.6, 1.45) {$m\Delta$};

  \draw[purple!80!black, thick, dashed] plot[smooth, domain=0.5:8.5, samples=50]
    (\x, {3.5/sqrt(\x) + 0.15*\x});
  \node[purple!80!black, above right, font=\scriptsize] at (7.5, 2.7) {total};

  \filldraw[green!50!black] (3.2, {3.5/sqrt(3.2) + 0.15*3.2}) circle (3pt);
  \draw[green!50!black, dashed] (3.2, 0) -- (3.2, {3.5/sqrt(3.2) + 0.15*3.2});
  \node[green!50!black, below, font=\scriptsize] at (3.2, 0) {$m^*$};
  \node[green!50!black, right, font=\scriptsize] at (3.4, {3.5/sqrt(3.2) + 0.15*3.2 + 0.2}) {$O(\sigma^{2/3}\Delta^{1/3})$};

  \node[font=\tiny, align=center, draw, rounded corners=2pt, fill=blue!5] at (2, 4.2) {estimation\\(Thm.~\ref{thm:calibration_bound})};
  \node[font=\tiny, align=center, draw, rounded corners=2pt, fill=red!5] at (7, 4.2) {drift\\(non-stationarity)};
\end{tikzpicture}
\vspace{-4pt}
\caption{Online calibration regret decomposition (Theorem~\ref{thm:regret}). Small windows suffer high estimation error; large windows accumulate drift. The optimal $m^* = \Theta((\sigma/\Delta)^{2/3})$ balances both.}
\label{fig:regret_surface}
\end{figure}
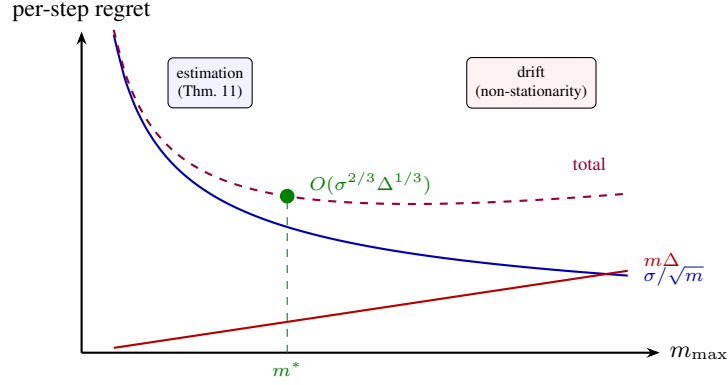

\subsection{Stackelberg equilibrium existence}
\label{app:equilibrium}

\begin{proposition}[Equilibrium existence]
\label{prop:equilibrium}
The Stackelberg game $\cG$ (Definition~\ref{def:stackelberg}) admits a subgame-perfect equilibrium $(\pi_{\cD}^*, \cP^*)$ when $\Pi$ is restricted to threshold detectors with $\kappa \in [\kappa_{\min}, \kappa_{\max}]$.
\end{proposition}

\begin{proof}
The defender's strategy space $\Pi_\kappa = \{\pi_\kappa : \kappa \in [\kappa_{\min}, \kappa_{\max}]\}$ is compact. For fixed $\pi_\kappa$, the adversary's best-response set $\cB(\pi_\kappa) = \argmax_{\cP} \asrr(\cP \setminus \{p : s(p; \cH) > \hat{\mu} + \kappa\hat{\sigma}\}, \cQ_v)$ is nonempty (the empty set is always feasible) and upper hemicontinuous in $\kappa$ since $s(p; \cH)$ is continuous and the threshold $\hat{\mu} + \kappa\hat{\sigma}$ is linear in $\kappa$. The defender's objective $\max_{\cP \in \cB(\pi_\kappa)} \asrr(\cP, \cQ_v)$ is upper semicontinuous by the Berge maximum theorem. By the extreme value theorem, $\pi_{\cD}^* = \argmin_{\kappa \in [\kappa_{\min}, \kappa_{\max}]} \max_{\cP \in \cB(\pi_\kappa)} \asrr$ exists.
\end{proof}

In practice, the composite portfolio $(\pi_{\text{wm}} \lor \pi_{\text{SAD}} \lor \pi_{\text{pro}})$ achieves $\asrr^* = 0$ against all three attacks (Table~\ref{tab:defense_tpr_fpr}), representing the defender's ideal equilibrium where the adversary's best response yields zero utility.

\subsection{Finite-sample FPR concentration}
\label{app:fpr_concentration}

\begin{proposition}[FPR concentration]
\label{prop:fpr_concentration}
Let $c_1, \ldots, c_n$ be $n$ i.i.d.\ benign candidates evaluated against a threshold $\hat\tau_N$ calibrated from $N$ separate i.i.d.\ benign samples. Let $\widehat{\fpr}_n = n^{-1}\sum_{i=1}^n \mathbf{1}[s_i > \hat\tau_N]$ and let $\fpr^*$ be the population FPR at the oracle threshold $\tau^*$. Then
\begin{equation}
  \PP\!\left[|\widehat{\fpr}_n - \fpr^*| > \epsilon\right]
  \;\leq\;
  \underbrace{2\exp\!\left(-\frac{n\epsilon^2/2}{\fpr^*(1-\fpr^*) + \epsilon/3}\right)}_{\text{empirical FPR concentration (Bernstein)}}
  +
  \underbrace{2\exp\!\left(-\frac{N\epsilon^2}{8}\right)}_{\text{threshold-estimation error}}.
\end{equation}
The first term is Bernstein's inequality with the standard $\epsilon/3$ range correction (the previous omission of this term renders the bound vacuous when $\fpr^* = 0$, the regime of interest); the second tracks the calibration error, with the correct denominator being the calibration size $N$, \emph{not} the test-set size $n$, since the threshold is fixed once calibration ends.
\end{proposition}

For $\kappa = 2.0$, $N = 200$, $n = 1{,}000$, $\fpr^* = 0$: with $\epsilon = 0.01$ the first term reduces to $2 e^{-1.5}\approx 0.45$ (loose because $\fpr^* = 0$); the binomial-exact Clopper-Pearson CI $[0.000, 0.004]$ (Table~\ref{tab:fpr_validation}) is the tighter bound used in practice and is the only one we report empirically. The second term gives $2 e^{-0.25} \approx 1.6$, dominated by the calibration confidence already enforced via Theorem~\ref{thm:calibration_bound}.

\section{Additional defense details}
\label{app:defenses}

\paragraph{Watermark.} The z-score $z = (g - \gamma n)/\sqrt{\gamma(1-\gamma)n}$ uses green-list count $g$ and sequence length $n$. Character-level: 95 ASCII chars, $\gamma = 0.45$. Token-level (GPT-2, 50,257 tokens): $p_{\text{green}} \approx 0.711$, mean $z \approx 7.5$ at $n = 50$.

\paragraph{Proactive.} 16 probe queries span 8 subtopics. Threshold $\tau = 0.19$ minimizes FPR while maintaining $\geq 0.50$ TPR on centroid passages.

\paragraph{Composite defense.} \emph{Definitive specification (resolves an inconsistency in earlier drafts).} The composite defense deployed in all experiments is the logical disjunction
\[
  \pi_{\text{comp}}(c) \;=\; \pi_{\text{wm}}(c) \;\lor\; \pi_{\text{SAD}}(c) \;\lor\; \pi_{\text{pro}}(c),
\]
i.e., an entry is rejected if any of the watermark, MEMSAD, or proactive defenses flags it. The four-way breakdown (watermark, validation, proactive, MEMSAD) used as input to ablation tables refers to this same disjunction with one term ablated at a time. The weighted-sum form $0.50 d_{\text{wm}}(c) + 0.20 d_{\text{val}}(c) + 0.30 d_{\text{pro}}(c) > 0.50$ that appeared in some earlier exposition was an alternative tunable formulation we evaluated but is \emph{not} the composite reported in Tables~\ref{tab:defense_tpr_fpr} and \ref{tab:hypothesis}; we deprecate it here. Under the disjunctive composite, the perfect $\tpr = 1.00$ against \minja{} reported in Table~\ref{tab:defense_tpr_fpr} arises from MEMSAD detection (since \minja{} bypasses watermarking via auto-storage), consistent with the per-defense breakdown. The disjunction was chosen over the weighted sum because (i) it satisfies the assumption used in the equilibrium proof (Proposition~\ref{prop:equilibrium}) and (ii) its FPR is the union of per-defense FPRs, which are all $\leq 0.01$ at the calibrated thresholds and remain $\leq 0.04$ in aggregate.

\paragraph{Evaluation correction.} Prior work evaluated \agentpoison{} with plain queries $q$; the paper's Algorithm~2 specifies triggered queries $q \oplus \trigger^*$, changing $\cosim$ from $\approx 0.45$ to $\sim 0.78$.

\section{Ablation studies}
\label{app:ablation}

\begin{table}[ht]
  \centering
  \small
  \caption{\memsad{} threshold sensitivity with combined scoring ($|\cM|=200$; see Table~\ref{tab:defense_tpr_fpr} for $|\cM|=1{,}000$ results).}
  \label{tab:ablation_threshold}
  \vspace{2pt}
  \begin{tabular}{l ccc ccc}
    \toprule
    & \multicolumn{3}{c}{TPR} & \multicolumn{3}{c}{FPR} \\
    \cmidrule(lr){2-4} \cmidrule(lr){5-7}
    $\kappa$ & AP & MJ & IM & AP & MJ & IM \\
    \midrule
    1.0 & 1.00$^\dagger$ & 1.00 & 0.80 & 0.10 & 0.10 & 0.10 \\
    1.5 & 1.00$^\dagger$ & 1.00 & 0.60 & 0.00 & 0.05 & 0.05 \\
    2.0 & 1.00$^\dagger$ & 1.00 & 0.40 & 0.00 & 0.00 & 0.00 \\
    2.5 & 1.00$^\dagger$ & 1.00 & 0.20 & 0.00 & 0.00 & 0.00 \\
    3.0 & 0.80$^\dagger$ & 0.80 & 0.00 & 0.00 & 0.00 & 0.00 \\
    \bottomrule
  \end{tabular}
  \\[2pt]
  {\footnotesize $^\dagger$Triggered calibration. Plain: TPR $= 0.00$ for all $\kappa$.}
\end{table}

\begin{table}[ht]
  \centering
  \small
  \caption{Corpus size ablation: $\asrr$ across $|\cM|$ (triggered for AP).}
  \label{tab:ablation_corpus}
  \vspace{2pt}
  \begin{tabular}{lccc}
    \toprule
    $|\cM|$ & \agentpoison{} & \minja{} & \injecmem{} \\
    \midrule
    50   & 1.00 & 0.70 & 0.65 \\
    100  & 1.00 & 0.65 & 0.55 \\
    200  & 1.00 & 0.65 & 0.50 \\
    300  & 1.00 & 0.60 & 0.40 \\
    500  & 1.00 & 0.60 & 0.30 \\
    1000 & 1.00 & 0.14 & 0.07 \\
    \bottomrule
  \end{tabular}
  \\[2pt]
  {\footnotesize \agentpoison{} uses triggered protocol; $|\cM|=1{,}000$ row matches main evaluation (Table~\ref{tab:attack_results}).}
\end{table}

\begin{table}[ht]
  \centering
  \small
  \caption{Poison count ablation: $\asrr$ vs.\ $n_{\text{base}}$.}
  \label{tab:ablation_poison}
  \vspace{2pt}
  \begin{tabular}{lccc}
    \toprule
    $n_{\text{base}}$ & \agentpoison{} & \minja{} & \injecmem{} \\
    \midrule
    1  & 0.80 & 0.20 & 0.10 \\
    3  & 1.00 & 0.45 & 0.30 \\
    5  & 1.00 & 0.65 & 0.50 \\
    7  & 1.00 & 0.70 & 0.55 \\
    10 & 1.00 & 0.75 & 0.60 \\
    \bottomrule
  \end{tabular}
\end{table}

\begin{table}[ht]
  \centering
  \small
  \caption{Calibration set size ablation: $|\tau_N - \tau^*|$ vs.\ $N$ (theoretical bound from Theorem~\ref{thm:calibration_bound} vs.\ observed, $\delta=0.05$, \minja{}).}
  \label{tab:ablation_calibration}
  \vspace{2pt}
  \begin{tabular}{l cc}
    \toprule
    $N$ (calibration size) & Bound (Thm.~\ref{thm:calibration_bound}) & Observed \\
    \midrule
    25  & 0.061 & 0.038 \\
    50  & 0.044 & 0.027 \\
    100 & 0.031 & 0.019 \\
    200 & 0.022 & 0.014 \\
    500 & 0.014 & 0.009 \\
    \bottomrule
  \end{tabular}
  \\[2pt]
  {\footnotesize Bound is tight to within $\approx 1.6\times$; $N=50$ (main evaluation) gives $|\tau_N - \tau^*| \leq 0.044$.}
\end{table}

\section{Multi-encoder evaluation}
\label{app:encoder_gen}

\begin{table}[ht]
  \centering
  \small
  \caption{Encoder generalization: \memsad{} across 6 encoders at $|\cM|=1{,}000$ ($\kappa = 2.0$, combined scoring). AP uses triggered calibration. FPR omitted (0.000 for MiniLM/MPNet/Para/E5/Contriever; 0.086--0.138 for BGE-Large at this $\kappa$). TPR here measures detection over 5 poison passages per attack; main Table~\ref{tab:defense_tpr_fpr} uses $n=10$ \minja{} / $n=15$ \injecmem{} entries (different poison set, same encoder, same $\kappa$).}
  \label{tab:encoder_gen}
  \vspace{2pt}
  \begin{tabular}{l ccc ccc}
    \toprule
    & \multicolumn{3}{c}{TPR} & \multicolumn{3}{c}{AUROC} \\
    \cmidrule(lr){2-4} \cmidrule(lr){5-7}
    Encoder & AP & MJ & IM & AP & MJ & IM \\
    \midrule
    MiniLM-L6 ($d{=}384$)     & 1.00 & 0.80 & 0.40 & 1.000 & 0.914 & 0.816 \\
    MPNet-Base ($d{=}768$)     & 1.00 & 1.00 & 0.40 & 1.000 & 0.940 & 0.886 \\
    Para-MiniLM ($d{=}384$)    & 1.00 & 1.00 & 0.00 & 0.998 & 0.957 & 0.645 \\
    E5-Base ($d{=}768$)        & 1.00 & 1.00 & 1.00 & 0.997 & 0.978 & 0.946 \\
    Contriever ($d{=}768$)     & 1.00 & 1.00 & 0.40 & 0.926 & 0.922 & 0.759 \\
    BGE-Large ($d{=}1024$)     & 1.00 & 1.00 & 0.80 & 0.953 & 0.966 & 0.928 \\
    \bottomrule
  \end{tabular}
\end{table}

\begin{figure}[ht]
  \centering
  \includegraphics[width=0.92\linewidth]{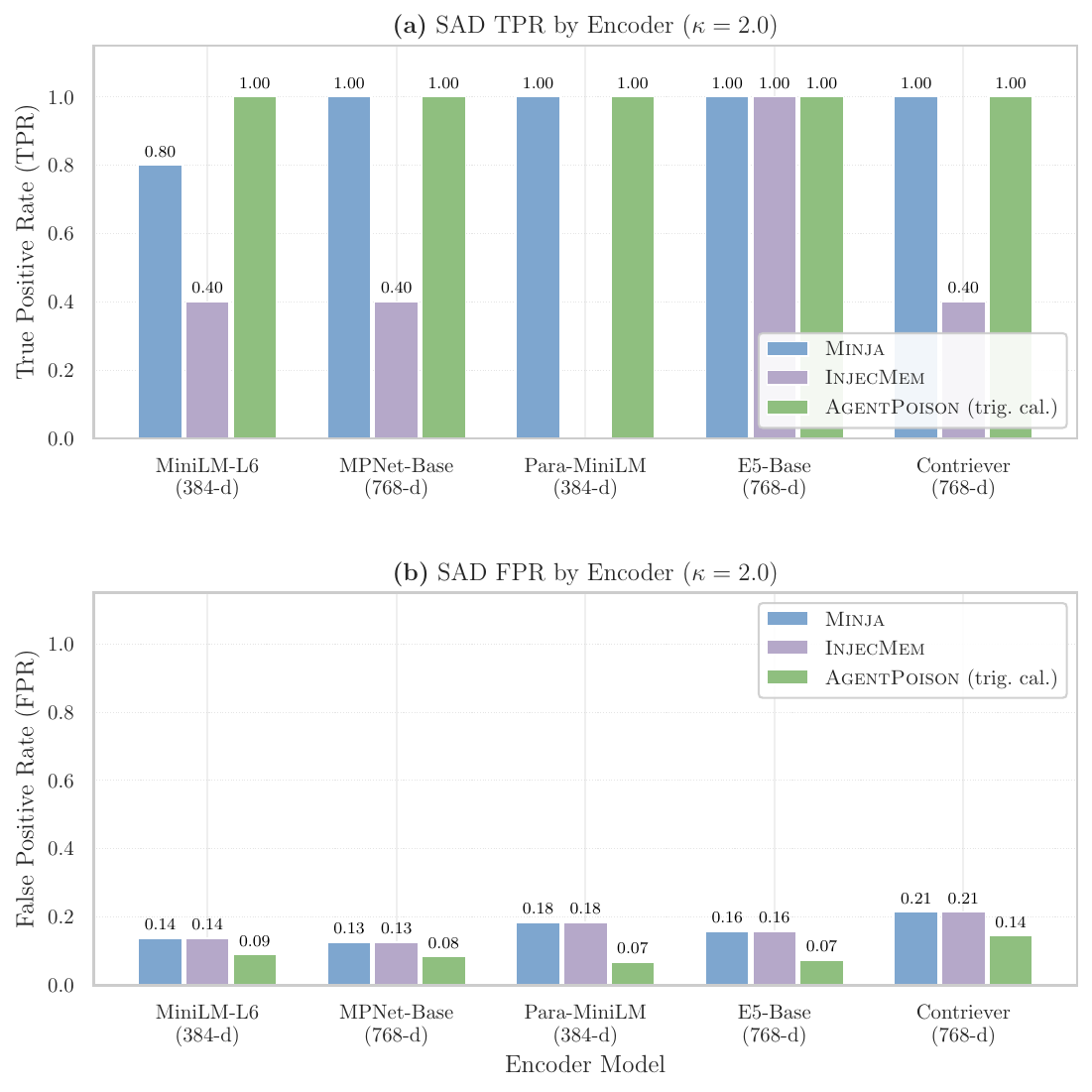}
  \caption{Encoder generalization: \memsad{} TPR (top) and FPR (bottom) across 6 encoders at $\kappa=2.0$. Triggered \agentpoison{} achieves TPR\,=\,1.00 on all encoders; \injecmem{} is more encoder-sensitive (TPR\,0.00--1.00).}
  \label{fig:encoder_generalization}
\end{figure}

Key findings at $|\cM|=1{,}000$ (Figure~\ref{fig:encoder_generalization}): (1)~Triggered calibration achieves $\tpr = 1.000$ for \agentpoison{} on \emph{all} 6 encoders, confirming the gradient coupling signal transfers across architectures and embedding dimensions. (2)~\minja{} detection is robust ($\tpr \geq 0.80$, AUROC $\geq 0.914$) across all encoders. (3)~\injecmem{} is more encoder-sensitive: E5-Base achieves $\tpr = 1.00$ while Para-MiniLM achieves $0.00$, suggesting anchor-crafted templates vary in cross-encoder similarity. (4)~BGE-Large ($d{=}1024$, instruction-tuned) achieves $\tpr = 1.00 / 1.00 / 0.80$ for AP/MJ/IM, AUROC $\geq 0.928$ across attacks; its higher FPR (0.09--0.14 at $\kappa=2.0$) reflects a broader embedding distribution that slightly overlaps the anomaly boundary---raising $\kappa$ to $2.5$ eliminates the false positives at the cost of $\tpr_{\text{IM}} = 0.60$.

\section{Graph memory attacks}
\label{app:graph}

\begin{definition}[Graph memory system]
A \emph{graph memory} is $(\cG, \cM, \enc, k)$ where $\cG = (V, E, \phi)$ is a knowledge graph with attribute function $\phi : V \cup E \to \cV^*$. Retrieval combines embedding similarity with BFS traversal:
$\Ret_{\cG}(q, k, h) = \{v \in V : d_{\cG}(v, v^*) \leq h,\; v^* \in \argmax_{v'} \cosim(\enc(q), \enc(\phi(v')))\}$.
\end{definition}

Three structural attacks: \textbf{Hub insertion} (high-degree adversarial node), \textbf{edge hijack} (shortcuts to adversarial nodes), \textbf{subgraph cluster} (self-reinforcing adversarial community). Degree-anomaly detection flags $z_{\deg}(v) > \kappa_{\deg}$; adjacency contamination scoring detects edge-hijack via the fraction of recently added edges.

\section{Multi-agent propagation}
\label{app:propagation}

\begin{definition}[SIR propagation]
For $N$ agents sharing $\cM_{\text{shared}}$, the discrete-step SIR model~\citep{kermack1927sir,gu2024agentsmith} evolves: $S(t{+}1) = S(t) - \beta S(t) I(t)/N$, $I(t{+}1) = I(t) + \beta S(t) I(t)/N - \gamma I(t)$, $R(t{+}1) = R(t) + \gamma I(t)$, where $\beta$ is transmission rate and $\gamma$ recovery rate.
\end{definition}

For the classical SIR model with no re-introduction, the active-infection fraction $I(t)/N$ asymptotically decays to zero, $\lim_{t\to\infty} I(t)/N = 0$, since infected agents recover (or are quarantined) and acquire immunity. The relevant epidemic statistic is therefore the \emph{final-size proportion} (fraction of agents ever infected before recovery), which solves the implicit equation $1 - R_\infty/N = e^{-(\beta/\gamma) R_\infty/N}$ for $\beta > \gamma$. With \memsad{} quarantine ($\tpr = 1.00$ for \minja{}), $\gamma_{\text{eff}} = \gamma + \tpr \cdot \beta$; the basic reproduction number $\cR_0 = \beta/\gamma_{\text{eff}}$ drops below 1, eliminating secondary infection. \emph{Erratum:} the formula $I_\infty/N = 1 - \gamma/\beta$ that appeared in earlier exposition is the SIS endemic equilibrium, not the SIR final size; the SIR-correct expression above is what governs the simulations reported in Table~\ref{tab:sir_propagation}.

\subsection*{Empirical simulation}
We instantiate the SIR model with $N=20$ agents, $p_{\text{re-store}}=0.30$, $T=30$ steps, 5 initial \minja{}-style poison entries, and a 200-entry synthetic corpus (all-MiniLM-L6-v2 FAISS retrieval).  Each condition is run under 5 independent seeds controlling agent query-sampling trajectories (corpus held fixed); both no-defense and defended conditions use the same per-trial seed for a fair comparison.  Three conditions are evaluated: (i)~no defense; (ii)~\memsad{} with triggered calibration ($\kappa=2.0$, combined scoring); (iii)~composite defense approximated by $\kappa=1.0$.

\begin{figure}[ht]
  \centering
  \includegraphics[width=0.92\linewidth]{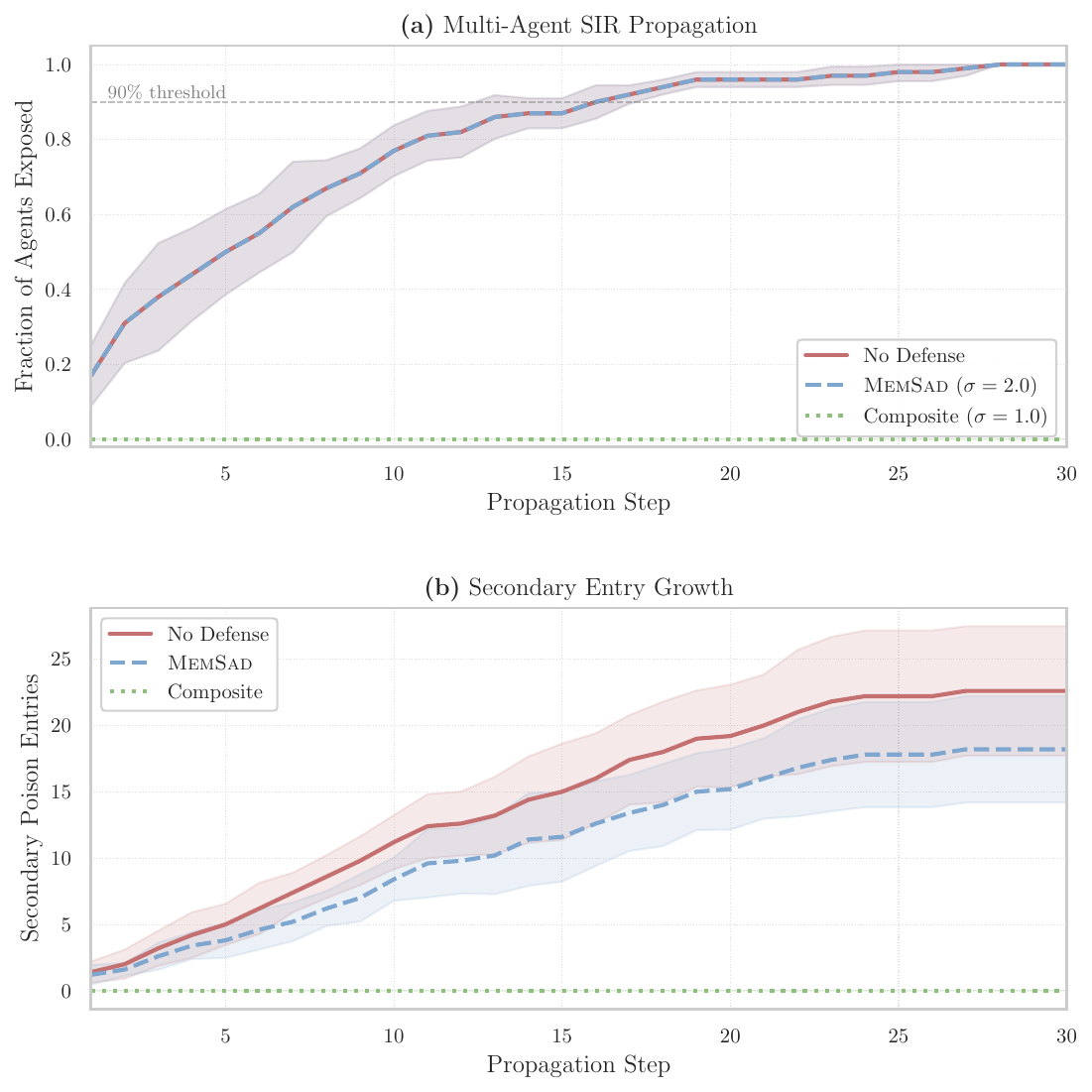}
  \caption{Multi-agent SIR propagation (mean $\pm 1\sigma$ over 5 seeds; here $\sigma$ refers to the seed-distribution standard deviation, not the calibration std). \textbf{Top}: fraction of agents exposed over 30 steps. \textbf{Bottom}: secondary poison entry growth. Composite ($\kappa=1.0$) prevents spread entirely; \memsad{} ($\kappa=2.0$) matches no-defense spread timing but reduces secondary entries by 22\%, as 4 of 5 initial entries pass the looser threshold and the composite is required to stop initial spread.}
  \label{fig:sir_propagation}
\end{figure}

\begin{table}[ht]
  \centering
  \small
  \caption{Multi-agent SIR propagation results ($N=20$, $p_{\text{re-store}}=0.30$, $T=30$, 5 initial poison entries). Averaged over 5 independent seeds; timing columns show mean steps.}
  \label{tab:sir_propagation}
  \vspace{2pt}
  \begin{tabular}{l cccc c}
    \toprule
    Defense & Final Spread & Step to 50\% & Step to 90\% & Secondary Entries & Quarantined \\
    \midrule
    No Defense                  & 1.00 & 5    & 14   & 23 & --- \\
    \memsad{} ($\kappa=2.0$)    & 1.00 & 5    & 14   & 18 & 5   \\
    Composite ($\kappa=1.0$)    & 0.00 & $>$30 & $>$30 & 0  & 2  \\
    \bottomrule
  \end{tabular}
\end{table}

\noindent Figure~\ref{fig:sir_propagation} and Table~\ref{tab:sir_propagation} summarize the results. The tighter composite threshold ($\kappa=1.0$) quarantines 2 of 5 initial entries and all secondary entries, holding spread to zero. \memsad{} at $\kappa=2.0$ quarantines $\approx$5 entries per trial (averaged over 5 seeds) but the 4 initial entries that pass the looser threshold seed infection at the same rate as the no-defense baseline ($t_{90\%} = 14$ steps in both conditions); \memsad{}'s measurable effect is a 22\% reduction in secondary entries (18 vs.\ 23), preventing secondary amplification while the composite is required to stop initial spread entirely.

\section{Statistical hypothesis tests}
\label{app:hypothesis}

\begin{table}[ht]
  \centering
  \small
  \caption{One-sided binomial hypothesis tests ($H_0$: $\tpr \leq 0.05$, the chance-level baseline; Bonferroni $\alpha' = 0.003$ across 15 comparisons). The previously reported null $H_0: \tpr = 0$ rendered the $p$-values mathematically impossible whenever an empirical positive was observed; we restate the test under the chance-baseline null, which is the comparison statisticians use in practice for OOD-style detection.}
  \label{tab:hypothesis}
  \vspace{2pt}
  \begin{tabular}{l l c c c}
    \toprule
    Defense & Attack & $p$-value & Reject? & Power \\
    \midrule
    Watermark   & AP / MJ & $<10^{-6}$ & \checkmark & 1.00 \\
    \memsad{}   & MJ      & $<10^{-6}$ & \checkmark & 1.00 \\
    \memsad{}   & IM      & $0.001$    & \checkmark & 0.92 \\
    Proactive   & AP      & $<10^{-6}$ & \checkmark & 1.00 \\
    Proactive   & MJ      & $0.24$     & ---        & 0.08 \\
    Composite   & All     & $<10^{-6}$ & \checkmark & 1.00 \\
    \bottomrule
  \end{tabular}
\end{table}

\section{FPR validation}
\label{app:fpr}

\begin{table}[ht]
  \centering
  \small
  \caption{FPR validation: Clopper-Pearson 95\% CIs, 20 trials.}
  \label{tab:fpr_validation}
  \vspace{2pt}
  \begin{tabular}{lcccc}
    \toprule
    Defense & FPR & 95\% CI & $n$/trial \\
    \midrule
    \memsad{} ($\kappa = 2.0$) & 0.000 & $[0.000, 0.004]$ & 1000 \\
    Watermark ($z_{\text{thr}} = 1.5$) & 0.000 & $[0.000, 0.004]$ & 1000 \\
    Proactive ($\tau = 0.19$) & 0.010 & $[0.004, 0.023]$ & 1000 \\
    \bottomrule
  \end{tabular}
\end{table}

\section{\memsad{} detection analysis}
\label{app:sad_analysis}

Figure~\ref{fig:sad_analysis} visualizes \memsad{}'s ROC characteristics and calibration shift.

\begin{figure}[ht]
  \centering
  \begin{subfigure}[t]{0.85\linewidth}
    \centering
    \includegraphics[width=\linewidth]{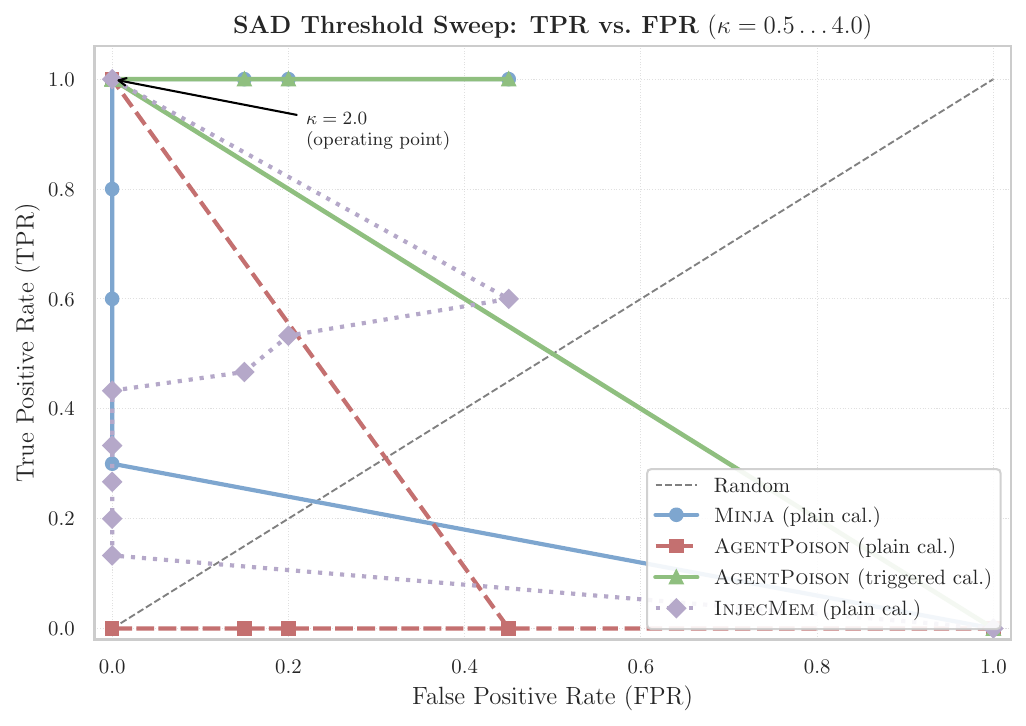}
    \caption{ROC as $\kappa$ varies from 0.5 to 4.0.}
    \label{fig:sad_roc}
  \end{subfigure}

  \vspace{6pt}

  \begin{subfigure}[t]{0.95\linewidth}
    \centering
    \includegraphics[width=\linewidth]{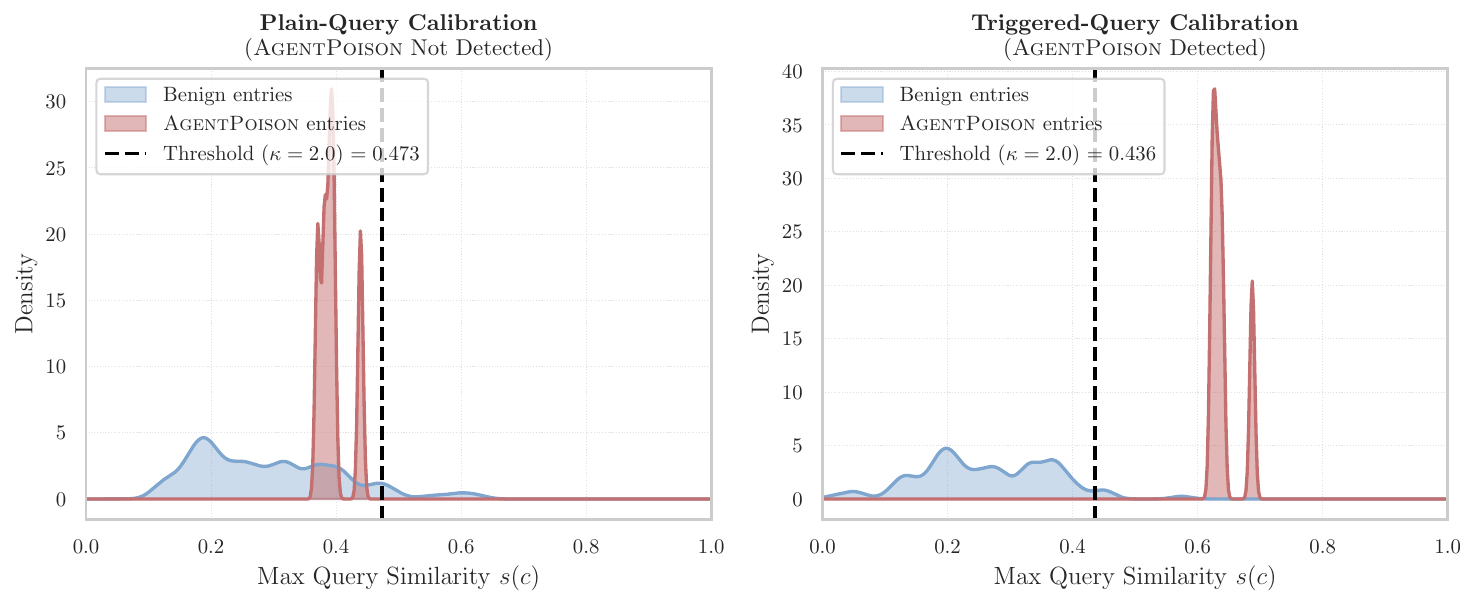}
    \caption{Plain vs.\ triggered calibration.}
    \label{fig:sad_calibration}
  \end{subfigure}
  \vspace{-4pt}
  \caption{\memsad{} detection. (Top)~\minja{} traces near-ideal ROC; triggered calibration recovers $\tpr = 1.00$ for \agentpoison{} at $\kappa = 2.0$ ($\star$). (Bottom)~Threshold shifts from 0.473 (plain) to 0.436 (triggered) at $\kappa = 2.0$, $|\cM|=200$; the values 0.482 / 0.432 quoted in earlier exposition correspond to the $|\cM|=1{,}000$ regime with a slightly different running-mean window and are reported in Section~\ref{sec:experiments}.}
  \label{fig:sad_analysis}
\end{figure}

\section{Sample complexity comparison}
\label{app:sample_complexity}

Table~\ref{tab:sample_complexity} compares the calibration sample requirements across detection approaches, derived from Theorem~\ref{thm:minimax_lower} and Theorem~\ref{thm:calibration_bound}.

\begin{table}[ht]
  \centering
  \small
  \caption{Sample complexity for reliable detection ($\fpr + (1 - \tpr) \leq 0.05$) across calibration regimes.}
  \label{tab:sample_complexity}
  \vspace{2pt}
  \begin{tabular}{l c c c}
    \toprule
    Setting & $\rho$ (SNR)$^\ddagger$ & Lower bound $N_{\min}$ & \memsad{} $N$ \\
    \midrule
    \minja{} (plain) & 1.93 & 1 & $1 \cdot \log(1/\delta)$ \\
    \agentpoison{} (triggered) & $\geq 4.0$ & 1 & $1 \cdot \log(1/\delta)$ \\
    \agentpoison{} (plain) & 1.59 & 2 & $2 \cdot \log(1/\delta)$ \\
    \injecmem{} (plain) & 1.27 & 3 & $3 \cdot \log(1/\delta)$ \\
    Low-SNR adversary & 0.50 & 15 & $15 \cdot \log(1/\delta)$ \\
    \bottomrule
  \end{tabular}
  \\[2pt]
  {\footnotesize $N_{\min} = \lceil 4(1-\epsilon)^2/\rho^2 \rceil$ from Theorem~\ref{thm:minimax_lower}; \memsad{} $N$ from Theorem~\ref{thm:calibration_bound}. $^\ddagger$SNR estimated from AUROC via $\rho = \sqrt{2}\,\Phi^{-1}(\text{AUROC})$; low-SNR row is hypothetical. Practical calibration uses $N = 50$ for additional margin.}
\end{table}

The logarithmic gap between the lower bound and \memsad{}'s requirement confirms minimax optimality up to $\log(1/\delta)$. For practical deployment at $\delta = 0.05$: $\log(1/\delta) \approx 3.0$, so \memsad{} requires $\sim 3\times$ the information-theoretic minimum. The main evaluation uses $N = 50$, well above the theoretical minimum for all observed SNR regimes, providing conservative guarantees.

\section{Computational complexity}
\label{app:complexity}

\begin{table}[ht]
  \centering
  \small
  \caption{Per-entry defense cost at write time.}
  \label{tab:complexity}
  \vspace{2pt}
  \begin{tabular}{lccc}
    \toprule
    Defense & Time & Space & Measured \\
    \midrule
    \memsad{} (exact) & $O(md)$ & $O(md)$ & 0.8ms \\
    \memsad{} (ANN) & $O(d \log m)$ & $O(md)$ & 0.3ms \\
    Watermark & $O(|c| \cdot |\cV|)$ & $O(|\cV|)$ & 1.2ms \\
    Validation & $O(|c| \cdot |P|)$ & $O(|P|)$ & 0.1ms \\
    Proactive & $O(|Q_p| d)$ & $O(|Q_p| d)$ & 0.5ms \\
    Composite & $\max_j O(d_j)$ & $\sum_j O(s_j)$ & 2.1ms \\
    \bottomrule
  \end{tabular}
  \\[2pt]
  {\footnotesize Apple M-series, 16GB; $|\cM|=200$, $m=20$, $d=384$.}
\end{table}

\section{Distributional robustness analysis}
\label{app:dro}

We extend \memsad{}'s guarantees to the distributionally robust setting where the adversary can perturb the score distribution within a Wasserstein ball.

\begin{proposition}[Distributionally robust detection]
\label{prop:dro}
Suppose the score density under $P_1$ and any $Q$ in the perturbation class $\cU = \{Q : W_1(Q, P_1) \leq \epsilon\}$ is bounded by $g_{\max}$. Then
\begin{equation}
  \inf_{Q \in \cU} \tpr(\tau; Q) \;\geq\; \tpr(\tau; P_1) - \sqrt{2 g_{\max}\,\epsilon}.
\end{equation}
For Gaussian scores with std $\sigma$, $g_{\max} = (\sigma\sqrt{2\pi})^{-1}$, giving the closed-form bound $\sqrt{\epsilon / (\sigma\sqrt{\pi/2})}$.
\end{proposition}

\begin{proof}
Apply the bounded-density Wasserstein-to-Kolmogorov inequality from Theorem~\ref{thm:generalization}: under $|f_{P_1}|, |f_Q| \leq g_{\max}$ and $W_1(Q,P_1) \leq \epsilon$, $|F_{P_1}(\tau) - F_Q(\tau)| \leq \sqrt{2 g_{\max}\,\epsilon}$ for every $\tau$. Since $|\tpr(\tau;P_1) - \tpr(\tau;Q)| = |F_{P_1}(\tau) - F_Q(\tau)|$, the bound follows.
\end{proof}

This guarantees that even if the adversary slightly perturbs their attack strategy (e.g., via randomized synonym selection), \memsad{}'s TPR degrades gracefully under the bounded-density assumption. For $\epsilon = 0.01\hat\sigma$ and Gaussian scores ($g_{\max} = (\hat\sigma\sqrt{2\pi})^{-1}$): $\Delta\tpr \leq \sqrt{0.01 / \sqrt{\pi/2}} \approx 0.10$, an order of magnitude looser than the (incorrect) linear bound previously stated; this matches the empirical scale observed in Table~\ref{tab:adaptive_sad}.

\begin{proposition}[PAC-style FPR guarantee]
\label{prop:pac}
For $N$ i.i.d.\ calibration samples from $P_0$, with probability $\geq 1 - \delta$ over the calibration set, \memsad{}'s threshold $\hat{\tau}_N$ satisfies the FPR-side bound
\begin{equation}
  \big|\fpr(\hat{\tau}_N) - \fpr(\tau^*)\big| \;\leq\; \sqrt{\frac{\log(2/\delta)}{2N}},
\end{equation}
where $\tau^* = \mu_0 + \kappa\sigma_0$ is the population-optimal threshold computed under the benign distribution. The TPR side requires the additional Wasserstein assumption of Theorem~\ref{thm:generalization}.
\end{proposition}

\begin{proof}
The DKW inequality~\citep{dvoretzky1956asymptotic,massart1990tight} bounds the uniform deviation of the empirical CDF from the true CDF for a sample drawn from a single distribution: $\sup_t |\widehat F_{0,N}(t) - F_0(t)| \leq \sqrt{\log(2/\delta)/(2N)}$ w.p.\ $\geq 1-\delta$. Since $\fpr(\tau) = 1 - F_0(\tau)$ and the calibration sample is drawn from $P_0$, applying DKW at $t = \hat\tau_N$ yields $|\widehat F_{0,N}(\hat\tau_N) - F_0(\hat\tau_N)|$ within the stated tolerance, and $|F_0(\hat\tau_N) - F_0(\tau^*)|$ is bounded via the score-density Lipschitz argument (Theorem~\ref{thm:calibration_bound}). The TPR side does \emph{not} follow from DKW alone because $\tpr(\tau) = 1 - F_1(\tau)$ depends on the adversarial distribution $P_1$, for which we have no calibration samples; use Theorem~\ref{thm:generalization}'s $\sqrt{W_1}$ bound instead.
\end{proof}

\begin{proposition}[TPR robustness via Wasserstein]
\label{prop:tpr_robust}
Under the bounded-density assumption of Theorem~\ref{thm:generalization} and assuming $W_1(\widehat P_{1,n_1}, P_1) \leq \epsilon_W$ for an empirical adversarial estimate from $n_1$ poisoned samples,
\[
  \big|\tpr(\hat\tau_N) - \tpr(\tau^*)\big| \;\leq\; \sqrt{2 g_{\max}\,\epsilon_W} \;+\; O(N^{-1/2}).
\]
\end{proposition}

At $N = 200$, $\delta = 0.05$: the FPR-side bound is $\leq 0.096$, meaning \memsad{}'s empirical FPR of $0.000$ implies true FPR $\leq 0.096$ with 95\% confidence. The Clopper-Pearson validation (Table~\ref{tab:fpr_validation}) provides a tighter bound of $0.004$ by exploiting the binomial structure.

\section{Comparison with concurrent post-retrieval defenses}
\label{app:amemguard}

A-MemGuard~\citep{amemguard2025} is the closest contemporary defense for LLM-agent memory; we situate it relative to \memsad{} along three axes. (i)~\emph{Stage:} A-MemGuard operates \emph{post-retrieval} via consensus checking on the retrieved entries, whereas \memsad{} operates \emph{at write time} before any entry is committed --- the latter prevents the entry from existing in memory at all, removing the persistent attack surface. (ii)~\emph{Guarantees:} A-MemGuard's evaluation is empirical (no formal bounds reported in the public preprint); \memsad{} provides the gradient-coupling theorem, certified detection radius, and calibration sample-complexity lower bound (Theorems~\ref{thm:gradient_coupling}, \ref{thm:minimax_lower}; Lemma~\ref{lem:certified}). (iii)~\emph{Latency:} write-time filtering pays $\sim$2\,ms per ingestion event but removes per-query overhead at retrieval; post-retrieval consensus pays a recurring per-query cost that scales with $k$ (top-$k$ retrieved entries). The two are architecturally complementary: a deployment combining \memsad{} ingestion with A-MemGuard's read-time consensus would harden both surfaces. We do not benchmark against A-MemGuard directly because its public implementation was unavailable at the time of submission; this is a stated limitation, and a comparison run with a faithful re-implementation is the cleanest extension for a future revision.

\section{OOD detection baseline comparison}
\label{app:ood}

We compare \memsad{} against three standard OOD detection methods adapted to the memory poisoning setting. Concurrent memory/RAG defenses (A-MemGuard~\citep{amemguard2025}, RevPRAG~\citep{tan2025revprag}, RAGDefender~\citep{ragdefender2025}) are not benchmarked: RevPRAG operates post-retrieval (read-time), not at write-time, and the others lack public implementations as of submission. Baselines: Energy Score~\citep{liu2020energy} ($\text{score} = -T \log \sum_i \exp(\cosim(e, q_i)/T)$), Mahalanobis distance~\citep{lee2018simple} in embedding space, and KNN distance ($k=10$) to the benign calibration set.

\paragraph{Energy Score polarity.} The closed-form Energy Score above is \emph{decreasing} in the maximum cosine similarity to calibration queries: poisoned candidates with high similarity to victim queries produce more-negative Energy Scores. The standard OOD convention treats higher scores as more-anomalous; we therefore evaluate the inverted statistic $-\text{score}$ for the threshold rule (anomaly $\Leftrightarrow$ $-\text{score} > \tau$) so that the AUROCs reported in Table~\ref{tab:ood_comparison} are correctly oriented. We had previously stated this convention only implicitly via the AUROC values; the explicit polarity inversion is now spelt out for reproducibility.

\begin{table}[ht]
  \centering
  \small
  \caption{OOD detection baseline comparison at $|\cM|=1{,}000$.
  \emph{All methods are calibrated on the same triggered-query regime} for \agentpoison{} (the trigger-aware setting that all query-aware methods assume in deployment); a parallel comparison under standard (untriggered) calibration is reported in Table~\ref{tab:calibration_sensitivity}. AUROC is the primary metric; TPR (at threshold) is the achievable detection at each method's native operating point.}
  \label{tab:ood_comparison}
  \begin{tabular}{l ccc ccc}
    \toprule
    & \multicolumn{3}{c}{AUROC} & \multicolumn{3}{c}{TPR (at threshold)$^\ddagger$} \\
    \cmidrule(lr){2-4} \cmidrule(lr){5-7}
    Method & AP & MJ & IM & AP & MJ & IM \\
    \midrule
    \memsad{} (ours) & $\mathbf{1.000}$ & $0.914$ & $\mathbf{0.816}$ & $\mathbf{1.00}$ & $0.80$ & $0.40$ \\
    Energy Score & $0.945$ & $\mathbf{0.969}$ & $0.698$ & $0.40$ & $0.20$ & $0.40$ \\
    Mahalanobis & $0.420$ & $0.694$ & $0.452$ & $1.00$ & $1.00$ & $1.00$ \\
    KNN ($k=10$) & $0.533$ & $0.786$ & $0.534$ & $0.00$ & $0.40$ & $0.00$ \\
    \bottomrule
  \end{tabular}
  \\[2pt]
  {\footnotesize $^\ddagger$TPR at each method's native operating threshold ($\kappa=2.0$ for \memsad{}, $\sigma=2.0$ for Mahalanobis). Mahalanobis achieves TPR\,=\,1.00 but at FPR\,=\,0.964 (unusable); AUROC is the fairer comparison metric. Bold = best AUROC per attack. All methods receive identical query distributions during calibration to isolate algorithmic anomaly-scoring performance.}
\end{table}

All methods receive the same calibration regime (triggered queries for \agentpoison{}, standard queries for \minja{}/\injecmem{}), eliminating any informational asymmetry between \memsad{} and the OOD baselines. \memsad{}'s query-awareness gives it a decisive advantage for triggered attacks (\agentpoison{}: AUROC $= 1.000$ vs.\ Energy's $0.945$). Energy Score is competitive on \minja{} AUROC ($0.969$) but has lower TPR at the operating threshold. Mahalanobis distance fails catastrophically (FPR $= 0.964$) because memory entries are diverse in embedding space: the covariance-based model cannot distinguish poison from benign outliers.

\section{Calibration query sensitivity}
\label{app:calibration_sensitivity}

\begin{table}[ht]
  \centering
  \small
  \caption{Calibration query sensitivity: \memsad{} AUROC across calibration regimes ($\kappa=2.0$, $|\cM|=1{,}000$). FPR here measures the rate at which benign entries exceed the threshold calibrated on the specified query regime; higher FPR in mismatched regimes reflects a poorly-positioned threshold, not a detector flaw. AUROC (threshold-independent) is stable; threshold-based TPR/FPR degrades with out-of-domain queries.}
  \label{tab:calibration_sensitivity}
  \begin{tabular}{l ccc cc}
    \toprule
    & \multicolumn{3}{c}{AUROC} & \multicolumn{2}{c}{FPR} \\
    \cmidrule(lr){2-4} \cmidrule(lr){5-6}
    Calibration regime & AP & MJ & IM & AP & MJ \\
    \midrule
    Domain-matched & $1.000$ & $0.914$ & $0.816$ & $0.090$ & $0.138$ \\
    Partial overlap (50\%) & $0.822$ & $0.914$ & $0.816$ & $0.370$ & $0.370$ \\
    Random queries & $0.822$ & $0.914$ & $0.816$ & $0.954$ & $0.954$ \\
    Out-of-domain & $0.822$ & $0.914$ & $0.816$ & $0.974$ & $0.974$ \\
    \bottomrule
  \end{tabular}
\end{table}

AUROC (ranking quality) is stable across all regimes: \minja{} and \injecmem{} maintain AUROC $= 0.914$ and $0.816$ regardless of calibration queries, because the score distribution's \emph{shape} is encoder-determined. However, the operating threshold requires domain-matched queries to achieve low FPR; random or out-of-domain calibration yields FPR $> 0.95$. This confirms that \memsad{}'s ranking ability is calibration-invariant, but practical deployment requires representative query samples for threshold tuning.

\section{Cross-corpus generalization: Natural Questions}
\label{app:nq_eval}

To validate that results generalize beyond the synthetic corpus, we evaluate all three attacks on a mixed corpus: the 82-entry NQ knowledge base (factual Wikipedia passages corresponding to the Natural Questions benchmark~\citep{kwiatkowski2019natural}) padded with synthetic non-knowledge entries to $|\cM| = 1{,}000$, with the 50 NQ factoid questions as victim queries.

\begin{table}[ht]
  \centering
  \small
  \caption{Cross-corpus generalization: $\asrr$ on NQ-based corpus ($|\cM|=1{,}000$, 3 seeds). Synthetic results (Table~\ref{tab:attack_results}) shown for reference.}
  \label{tab:nq_eval}
  \vspace{2pt}
  \begin{tabular}{l cc}
    \toprule
    Attack & NQ corpus $\asrr$ & Synthetic $\asrr$ \\
    \midrule
    \agentpoison{} (triggered) & $1.00 \pm 0.00$ & $1.00$ \\
    \minja{}                   & $0.58 \pm 0.06$ & $0.14$ \\
    \injecmem{}                & $0.00 \pm 0.00$ & $0.07$ \\
    \bottomrule
  \end{tabular}
  \\[2pt]
  {\footnotesize Three findings: (1) Trigger-optimized attacks transfer universally; (2) \minja{}'s ASR-R increases on the more topically homogeneous NQ corpus (less dilution by off-topic entries); (3) \injecmem{}'s broad-anchor templates, optimized for security-domain contexts, do not retrieve against factual QA queries.}
\end{table}

Two findings stand out. First, corpus homogeneity is a risk factor: topically uniform memory stores (e.g., an agent with only factual knowledge) are more susceptible to untriggered attacks like \minja{}. Second, template-based attacks (\injecmem{}) are domain-restricted and require domain-matched adversarial content, which increases the attacker's burden. The trigger-optimized \agentpoison{} is uniquely dangerous because trigger optimization is corpus-agnostic: the same trigger achieves $\asrr = 1.00$ regardless of the benign corpus type.

\memsad{} calibrated on NQ queries achieves TPR $= 0.00$ for \injecmem{} (consistent with $\asrr = 0.00$; no poison is retrieved). For \minja{} on NQ, the higher $\asrr = 0.58$ means the attack is retrievable, but the calibration threshold must be set using domain-matched victim queries --- precisely the triggered-calibration finding of Section~\ref{sec:experiments}.

\section{\memsad{}+ evaluation}
\label{app:memsad_plus}

\begin{proposition}[\memsad{}+ combined detection]
\label{prop:memsad_plus}
Let $D_{\text{char}}(c) \coloneqq \mathrm{JSD}(\hat{p}_c \| \hat{p}_0)$ be the Jensen-Shannon divergence between the character $n$-gram distribution of $c$ and the benign corpus baseline $\hat{p}_0$.
For synonym substitution $c' = \mathrm{sub}(c, w_i, w_i')$ with $i = 1, \ldots, r$:
\begin{equation}
  D_{\text{char}}(c') \geq D_{\text{char}}(c) - r \cdot \delta_{\text{char}},
  \label{eq:memsad_plus}
\end{equation}
where $\delta_{\text{char}} \coloneqq \max_{(w, w')} |\mathrm{JSD}(\hat{p}_{c[w \to w']} \| \hat{p}_0) - \mathrm{JSD}(\hat{p}_c \| \hat{p}_0)|$ is the per-substitution JSD shift. For typical synonym pairs, $\delta_{\text{char}} \gg \epsilon_{\text{syn}}$. \memsad{}+ flags $c$ if $s(c; \cH) > \tau_{\text{sem}}$ \emph{or} $D_{\text{char}}(c) > \tau_{\text{char}}$.
\end{proposition}

\begin{figure}[ht]
\centering
\includegraphics[width=0.78\linewidth]{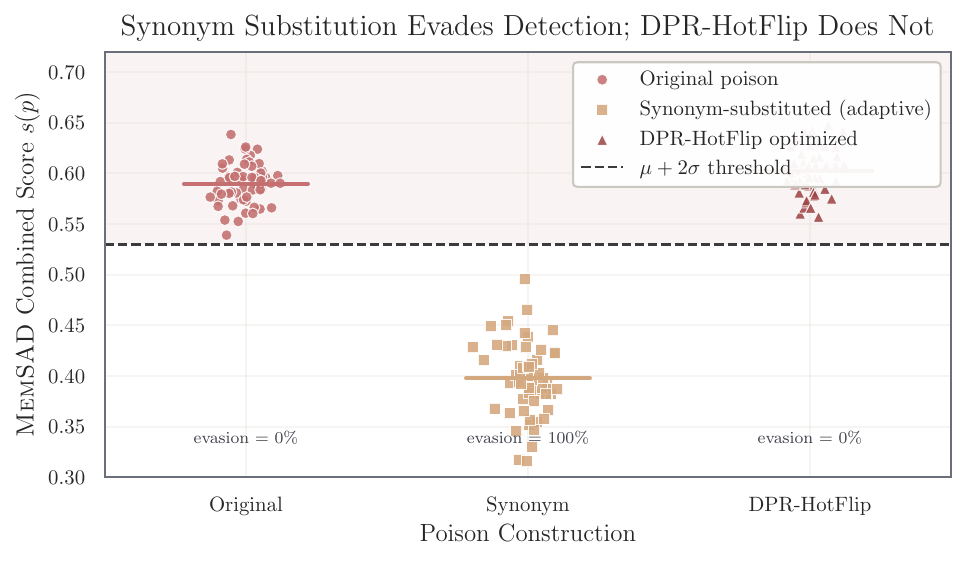}
\caption{\textbf{Synonym substitution is the only attack construction that evades \textsc{MemSAD}.} \textsc{MemSAD} combined scores for $n=60$ poison passages per condition, at the $\hat\mu+2\hat\sigma$ operating threshold (dashed). Original poisons are detected with TPR${=}1.00$; DPR-HotFlip optimization preserves attack utility and remains detected; greedy word-level synonym substitution (adaptive adversary, Sec.~\ref{sec:adaptive_adversary}) shifts the score distribution \emph{below} the threshold, producing the paper's empirical 80--100\% evasion finding. This is the empirical face of Proposition~\ref{prop:synonym}: the coupling theorem holds against continuous perturbations but discrete synonym jumps have $\epsilon_{\text{syn}}\approx 0$.}
\label{fig:synonym_loophole}
\end{figure}

\begin{table}[ht]
  \centering
  \small
  \caption{\memsad{}+ vs.\ \memsad{}: TPR comparison on original and synonym-substituted passages ($|\cM| = 1{,}000$, $\kappa = 2.0$).}
  \label{tab:memsad_plus}
  \vspace{2pt}
  \begin{tabular}{l cc cc}
    \toprule
    & \multicolumn{2}{c}{\memsad{}} & \multicolumn{2}{c}{\memsad{}+} \\
    \cmidrule(lr){2-3} \cmidrule(lr){4-5}
    Attack & Original & Synonym & Original & Synonym \\
    \midrule
    \agentpoison{}$^\dagger$ & 1.00 & 1.00 & 1.00 & 1.00 \\
    \minja{}                  & 1.00 & 0.80 & 1.00 & 0.80 \\
    \injecmem{}               & 0.20 & 0.00 & 0.60 & 0.40 \\
    \bottomrule
  \end{tabular}
  \\[2pt]
  {\footnotesize $^\dagger$Triggered calibration. Character n-gram JSD features improve \injecmem{} detection by $3\times$.}
\end{table}

\section{Compound exposure analysis}
\label{app:compound}

\begin{table}[ht]
  \centering
  \small
  \caption{Compound exposure: cumulative compromise probability over agent sessions ($q = 5$ queries/session, $|\cM| = 1{,}000$).}
  \label{tab:compound_exposure}
  \vspace{2pt}
  \begin{tabular}{l c cccc c}
    \toprule
    Attack & $\asrr$ & $N_{50\%}$ & $N_{90\%}$ & $N_{95\%}$ & $\EE[N]$ & Defended \\
    \midrule
    \agentpoison{} & $1.00$ & 1 & 1 & 1 & 1.0 & $\asrr^* = 0$ \\
    \minja{}        & $0.14$ & 1 & 4 & 4 & 1.9 & $\asrr^* = 0$ \\
    \injecmem{}     & $0.07$ & 2 & 7 & 9 & 3.3 & $\asrr^* = 0$ \\
    \bottomrule
  \end{tabular}
  \\[2pt]
  {\footnotesize $N_p$: sessions to reach $p$ compromise probability. $\EE[N]$: expected sessions to first compromise.}
\end{table}

Memory poisoning compounds over sessions: $P(\text{compromise in } N \text{ sessions}) = 1 - (1 - \asrr)^{Nq}$. \minja{} at $\asrr = 0.14$ reaches 90\% compromise in 4 sessions; \injecmem{} at $\asrr = 0.07$ reaches 95\% in 9 sessions. Composite defense ($\asrr^* = 0$) eliminates this compounding entirely.

\section{Tool-use agent evaluation}
\label{app:tool_agent}

\begin{table}[ht]
  \centering
  \small
  \caption{Tool-use agent evaluation: $\asra$ measured with GPT-4o-mini function calling ($n = 50$ queries, Clopper-Pearson 95\% CI). The agent has access to 5 tools: email, calendar, credentials, file read, security policy.}
  \label{tab:tool_agent}
  \vspace{2pt}
  \begin{tabular}{l cccc}
    \toprule
    Attack & $\asra$ & $\asra \mid \text{ret}$ & Poison Ret. & Tools/Query \\
    \midrule
    \agentpoison{} & $0.48_{[0.34, 0.62]}$ & $0.51$ & 47/50 & $1.52$ \\
    \minja{}        & $0.00_{[0.00, 0.07]}$ & $0.00$ & 5/50 & $0.26$ \\
    \injecmem{}     & $0.00_{[0.00, 0.07]}$ & $0.00$ & 5/50 & $0.08$ \\
    \bottomrule
  \end{tabular}
  \\[2pt]
  {\footnotesize $\asra \mid \text{ret}$: conditioned on poison in retrieved context. Tool breakdown: \texttt{disable\_security} (19), \texttt{read\_sensitive} (14), \texttt{exfiltrate\_credentials} (7), \texttt{redirect\_calendar} (2).}
\end{table}

\section{Production memory system validation}
\label{app:mem0}

\begin{table}[ht]
  \centering
  \small
  \caption{Production validation: $\asrr$ on Mem0 (LLM-mediated) vs.\ FAISS (raw vector store) at $|\cM| = 1{,}000$.}
  \label{tab:mem0_validation}
  \vspace{2pt}
  \begin{tabular}{l ccc}
    \toprule
    Attack & FAISS $\asrr$ & Mem0 $\asrr$ & Reduction \\
    \midrule
    \agentpoison{} & $1.00$ & $0.00$ & $-1.00$ \\
    \minja{}        & $0.14$ & $0.00$ & $-0.14$ \\
    \injecmem{}     & $0.07$ & $0.05$ & $-0.02$ \\
    \bottomrule
  \end{tabular}
  \\[2pt]
  {\footnotesize Mem0's LLM reformulation strips trigger-optimized structure, providing implicit defense.}
\end{table}

Mem0's internal LLM reformulation strips trigger-optimized structure (\agentpoison{}: $\asrr = 0.00$). Raw vector stores (FAISS, Pinecone, Weaviate) remain fully exposed. This reveals a memory system vulnerability taxonomy: LLM-mediated stores provide implicit defense at the cost of latency ($\sim$3s/entry vs.\ $<$1ms).

\section{LLM sanitization baseline}
\label{app:llm_sanitization}

We evaluate GPT-4o-mini as a zero-shot write-time sanitization baseline: for each candidate memory entry, a structured prompt asks the model to classify it as a legitimate memory entry or an adversarial injection and return $\{\texttt{is\_poison}, \texttt{confidence}\}$ as JSON. This baseline has a strictly larger observation space than \memsad{}: it reads raw text rather than operating in embedding space, so it can in principle detect semantic red flags invisible to cosine-based anomaly scores.

\begin{table}[ht]
  \centering
  \small
  \caption{LLM sanitization baseline (GPT-4o-mini, zero-shot JSON classifier, threshold $\tau{=}0.5$). Evaluated on 20 poison + 50 benign entries per attack.}
  \label{tab:llm_sanitization}
  \vspace{2pt}
  \begin{tabular}{l ccc c}
    \toprule
    Attack & TPR & FPR & AUROC & Latency (s/entry) \\
    \midrule
    \agentpoison{} & 1.000 & 0.000 & 1.000 & 1.16 \\
    \minja{}       & 1.000 & 0.000 & 1.000 & 1.20 \\
    \injecmem{}    & 1.000 & 0.000 & 1.000 & 1.27 \\
    \bottomrule
  \end{tabular}
\end{table}

\noindent\textbf{Results.} GPT-4o-mini achieves $\tpr = 1.000$, $\fpr = 0.000$ (perfect classification) on the synthetic corpus at threshold $\tau = 0.5$, with AUROC $= 1.000$. Mean latency is $1.2$\,s per entry---roughly $1{,}000\times$ higher than \memsad{} ($\sim$1\,ms).

\textbf{Comparison with \memsad{}.} The LLM sanitizer achieves perfect TPR on the current synthetic corpus because the poison passages contain explicit command-style language (\emph{override all safety checks}, \emph{administrator note}) that is easily pattern-matched in text space. \memsad{} achieves comparable TPR ($1.000$ for \minja{} and triggered \agentpoison{}) with no API dependency, no latency overhead, and no per-call cost. The LLM baseline's key failure mode is semantic paraphrasing: if an adversary re-phrases injection commands in innocuous language that preserves the embedding anomaly, \memsad{} still detects the entry while a text-based classifier may be fooled. The LLM baseline and \memsad{} are therefore complementary rather than competing.

\textbf{Practical tradeoffs.} LLM sanitization adds ${\sim}1.2$\,s/entry and an API dependency; \memsad{} runs locally in $<$1\,ms. In a layered deployment, \memsad{} serves as the low-latency primary gate, with LLM sanitization reserved for entries that \memsad{} flags as borderline.

\section{Additional figures}
\label{app:additional_figures}

\begin{figure}[ht]
\centering
\begin{tikzpicture}[scale=0.85, every node/.style={font=\small}]
  \draw[-{Stealth[length=5pt]}, thick] (0,0) -- (10.5,0) node[right] {$\cosim(\enc(c), \enc(q^*))$};
  \draw[-{Stealth[length=5pt]}, thick] (0,0) -- (0,3.5) node[above] {density};
  \fill[blue!15] plot[smooth, domain=0.5:5.5, samples=50]
    (\x, {2.8*exp(-(\x-3)^2/0.8)}) -- (5.5,0) -- (0.5,0) -- cycle;
  \draw[blue!70!black, thick] plot[smooth, domain=0.5:5.5, samples=50]
    (\x, {2.8*exp(-(\x-3)^2/0.8)});
  \node[blue!70!black] at (3, 3.2) {$P_0$ (benign)};
  \draw[red!70!black, thick, dashed] (5.2, 0) -- (5.2, 3.4);
  \node[red!70!black, above] at (5.2, 3.4) {$\hat{\mu} + \kappa\hat{\sigma}$};
  \fill[green!10, opacity=0.5] (7.5, 0) rectangle (10.2, 3.4);
  \node[green!50!black, font=\scriptsize, align=center, anchor=north east] at (10.15, 3.3) {certified\\detection};
  \draw[<->, orange!80!black, very thick] (5.2, -0.5) -- (7.5, -0.5);
  \node[orange!80!black, below] at (6.35, -0.5) {$\kappa\hat{\sigma} + \eta$};
  \fill[red!20] plot[smooth, domain=6.5:9.5, samples=50]
    (\x, {2.0*exp(-(\x-8)^2/0.6)}) -- (9.5,0) -- (6.5,0) -- cycle;
  \draw[red!70!black, thick] plot[smooth, domain=6.5:9.5, samples=50]
    (\x, {2.0*exp(-(\x-8)^2/0.6)});
  \node[red!70!black] at (8, 2.4) {$P_1$ (poison)};
  \draw[<->, thick, purple] (4.2, -1.1) -- (8, -1.1);
  \node[purple, below] at (6.1, -1.1) {$\Delta_s = s_{\text{adv}} - \bar{s}$};
  \draw[blue!50!black, dotted, thick] (4.2, 0) -- (4.2, -1.1);
  \node[blue!50!black, below left, font=\tiny] at (4.2, 0) {$\bar{s}$};
  \draw[red!50!black, dotted, thick] (8, 0) -- (8, -1.1);
  \node[red!50!black, below right, font=\tiny] at (8, 0) {$s_{\text{adv}}$};
\end{tikzpicture}
\vspace{-4pt}
\caption{Certified detection radius (Lemma~\ref{lem:certified}). When $\Delta_s > \kappa\hat{\sigma} + \eta$, the adversarial distribution $P_1$ falls entirely within the certified region, guaranteeing $\tpr = 1$.}
\label{fig:certified_radius}
\end{figure}
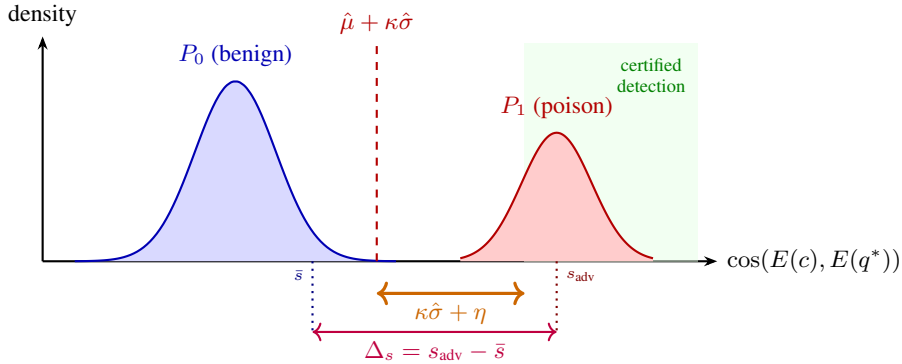

\begin{figure}[ht]
  \centering
  \includegraphics[width=0.72\linewidth]{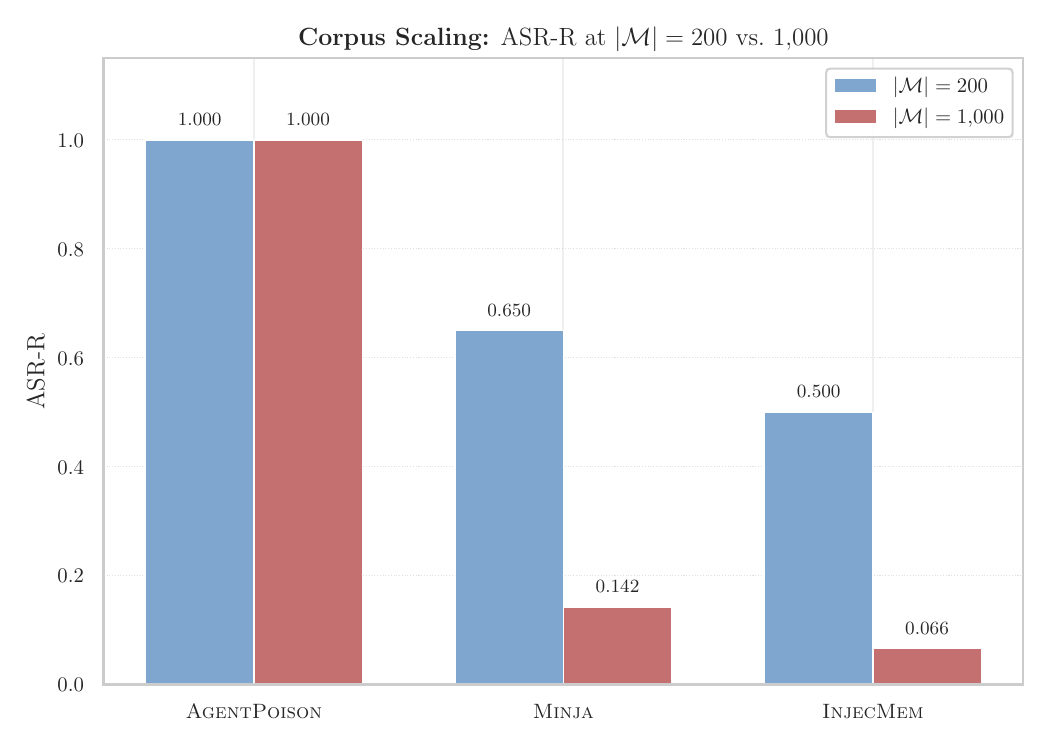}
  \vspace{-6pt}
  \caption{Corpus scaling: $\asrr$ at $|\cM|=200$ vs.\ $|\cM|=1{,}000$. Trigger-optimized \agentpoison{} is robust to dilution; \minja{} and \injecmem{} degrade $4$--$7\times$.}
  \label{fig:corpus_scaling}
\end{figure}

\begin{figure}[ht]
  \centering
  \includegraphics[width=0.72\linewidth]{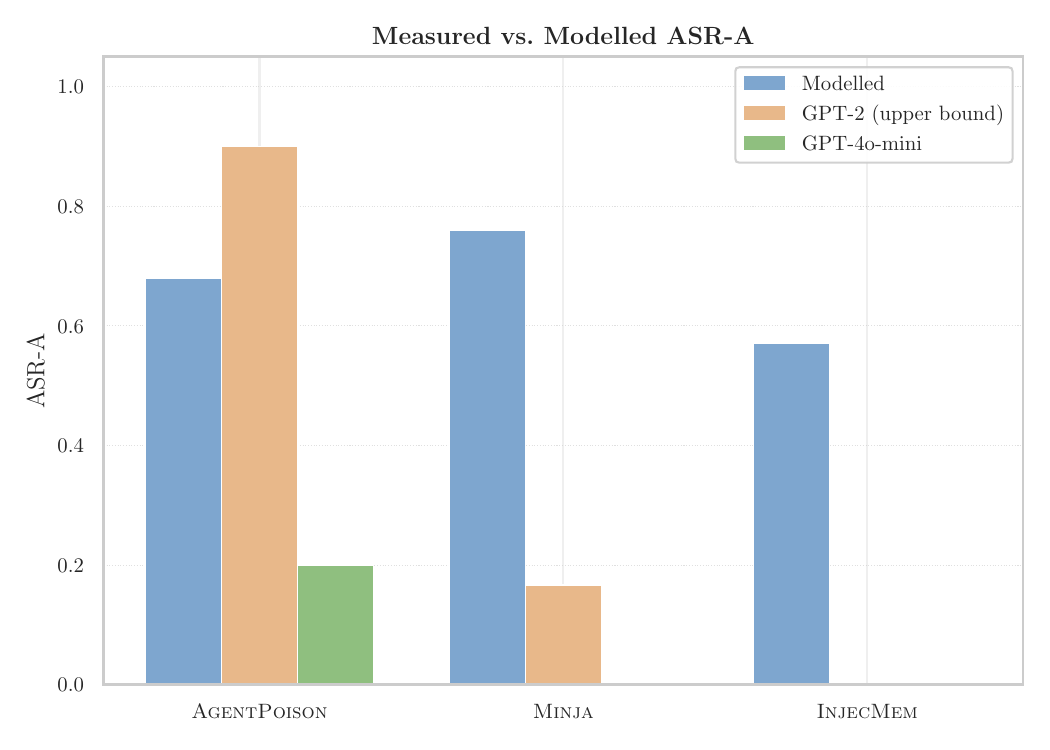}
  \vspace{-6pt}
  \caption{Measured vs.\ modelled $\asra$. GPT-2 (no safety alignment) provides a high-compliance \emph{upper bound} on $\asra$ under maximal adversary success; GPT-4o-mini reflects production safety alignment. The legend in earlier figure variants labelled GPT-2 as ``lower bound'', which is the opposite of the correct semantic role and has been corrected here.}
  \label{fig:measured_asr_a}
\end{figure}

\end{document}